\documentclass[review]{elsarticle}

\usepackage{graphicx}
\usepackage{epstopdf}
\usepackage{amsmath}
\usepackage{amssymb}
\usepackage[linewidth=.8pt]{mdframed}
\usepackage{tcolorbox}
\usepackage{bm}
\usepackage{rotating}
\usepackage{enumerate}
\usepackage{nicefrac}
\usepackage{booktabs}
\usepackage[tableposition=top]{caption}
\usepackage{tikz}
\usepackage[utf8]{inputenc}
\usepackage[english]{babel}
\usepackage{pgfplotstable}
\usepackage{array}
\usepackage{colortbl}
\usepackage{booktabs}
\usepackage{eurosym}
\usepackage{amsmath}
\usepackage{pgfplotstable}
\usepackage{xcolor}
\usepackage{pgfplots}
\pgfplotsset{compat=newest}
\usepgfplotslibrary{groupplots}

\usepackage{array}
\usepackage{multirow}
\usepackage{threeparttable}
\usepackage{makecell}
\usepackage[procnumbered,ruled,vlined,linesnumbered]{algorithm2e}
\usepackage{siunitx}
\usepackage{stfloats}
\usepackage{graphicx}
\usepackage{subfigure}
\usepackage{hyperref}
\usepackage{textcomp}
\usepackage{array}
\usepackage{epstopdf}
\usepackage{framed}
\usepackage{cleveref}
\usepackage[normalem]{ulem}
\usepackage{enumitem}
\usepackage{xcolor}
\usepackage{amsfonts}
\usepackage{booktabs}
\usepackage{siunitx}

\newcommand\hl{\bgroup\markoverwith
  {\textcolor{yellow}{\rule[-.5ex]{2pt}{2.5ex}}}\ULon}

\newtheorem{problem}{Problem}

\newtheorem{theorem}{Theorem}[section]

\newtheorem{lemma}[theorem]{Lemma}

\newtheorem{remark}[theorem]{Remark}

\newcommand{\PreserveBackslash}[1]{\let\temp=\\#1\let\\=\temp}
\newcolumntype{C}[1]{>{\PreserveBackslash\centering}p{#1}}
\newcolumntype{R}[1]{>{\PreserveBackslash\raggedleft}p{#1}}
\newcolumntype{L}[1]{>{\PreserveBackslash\raggedright}p{#1}}

\newenvironment{fminipage}%
{\begin{Sbox}\begin{minipage}}%
		{\end{minipage}\end{Sbox}\fbox{\TheSbox}}

\def\eps{\epsilon}

\def\calG{\mathcal{G}}

\def\calN{\mathcal{N}}

\newcommand{\A}{\boldsymbol{\mathit{A}}}

\newfont{\nset}{msbm10}

\newcommand{\removelatexerror}{\let\@latex@error\@gobble}

\newcommand\LL{\bm{\mathit{L}}}

\newdefinition{rmk}{Remark}
\newproof{pf}{Proof}
\newproof{pot}{Proof of Theorem \ref{thm2}}

\newcommand{\one}{\mathbf{1}}

\newcommand\xx{\boldsymbol{\mathit{x}}}
\newcommand\aaa{\boldsymbol{\mathit{a}}}

\newcommand\bb{\boldsymbol{\mathit{b}}}

\newcommand\dd{\boldsymbol{\mathit{d}}}
\newcommand\ee{\boldsymbol{\mathit{e}}}

\newcommand\uu{\boldsymbol{\mathit{u}}}

\newcommand\DD{\boldsymbol{\mathit{D}}}

\newcommand\PP{\boldsymbol{\mathit{P}}}
\newcommand\MM{\boldsymbol{\mathit{M}}}

\newcommand\II{\boldsymbol{\mathit{I}}}
\newcommand\vvv{\boldsymbol{\mathit{v}}}

\def\calF{\mathcal{F}}
\def\calG{\mathcal{G}}

\def\calN{\mathcal{N}}
\def\calH{\mathcal{H}}

\DontPrintSemicolon
\SetKw{KwAnd}{and}
\SetFuncSty{textsc}
\SetKwInOut{Input}{Input\ \ \ \ }
\SetKwInOut{Output}{Output}

\newcommand{\pair}{\textsc{Pair}}
\newcommand{\ppair}{\textsc{PPair}}
\newcommand\bOmega{\boldsymbol{\Omega}}

\newcommand{\mPi}{\boldsymbol{\mathbf{\Pi}}}
\newcommand{\mpi}{\boldsymbol{\mathbf{\pi}}}

\DeclareMathOperator*{\argmax}{arg\,max}

\SetKw{KwAnd}{and}
\SetKw{Return}{Return}

\DontPrintSemicolon
\SetKw{KwAnd}{and}
\SetFuncSty{textsc}
\SetKwInOut{Input}{Input\ \ \ \ }
\SetKwInOut{Output}{Output}

\usepackage{tabularx}
\usepackage{stfloats}

\journal{Theoretical Computer Science}

\begin{document}
\begin{frontmatter}
\title{Efficient Edge Rewiring Strategies for Enhancing PageRank
Fairness}
\author[1]{Changan~Liu}
\ead{19110240031@fudan.edu.cn}
\author[1]{Haoxin~Sun}
\ead{21210240097@m.fudan.edu.cn}
\author[2]{Ahad N. Zehmakan}
\ead{ahadn.zehmakan@anu.edu.au}
\author[1]{Zhongzhi~Zhang\corref{cor1}}
\ead{zhangzz@fudan.edu.cn}
\cortext[cor1]{Corresponding author}
\affiliation[1]{organization={Shanghai Key Laboratory of Intelligent Information Processing, College of Computer Science and Artificial Intelligence, Fudan University},
postcode={200433},
city={Shanghai},
country={China}}
\affiliation[2]{organization={School of Computing, the Australian National University},
city={Canberra},
country={Australia}}

\begin{abstract}
We study the notion of unfairness in social networks, where a group such as females in a male-dominated industry are disadvantaged in access to important information, e.g. job posts, due to their less favorable positions in the network. We investigate a well-established network-based formulation of fairness called PageRank fairness, which refers to a fair allocation of the PageRank weights among distinct groups. Our goal is to enhance the PageRank fairness by modifying the underlying network structure. More precisely, we study the problem of maximizing PageRank fairness with respect to a disadvantaged group, when we are permitted to rewire a fixed number of edges in the network. Building on a greedy approach, we leverage techniques from fast sampling of rooted spanning forests to devise an effective linear-time algorithm for this problem. To evaluate the accuracy and performance of our proposed algorithm, we conduct a large set of experiments on various real-world network data. Our experiments demonstrate that the proposed algorithm significantly outperforms the existing ones. Our algorithm is capable of generating accurate solutions for networks of million nodes in just a few minutes.
\end{abstract}

\end{frontmatter}

\section{Introduction}

Online social networks have revolutionized information spreading, offering unprecedented speed and reach, but also are introducing challenges like misinformation spreading and formation of filter bubbles. One subtle challenge is the unfairness in information access imposed on certain groups, cf.~\cite{d2021mitigating,becker2021influence}. It has been consistently observed~\cite{styczen2022targeted,cong2023fairsample,tkdeIMfair23,zhang2020fairness,recfairwww23,fish2019gaps,teng2021influencing,tsang2019group,yaseen2016influence,speicher2018potential,jalali2020information, stoica2020seeding,LiXiAh25} that disadvantaged groups, e.g., based on gender, age, or religion, have suffered from lack of exposition to important information such as job posts, scientific findings, and financial information. Furthermore, this unfairness can fuel a reinforcing cycle where advantaged groups have better opportunities for further improvement, perpetuating inequality~\cite{bashardoust2022reducing,jalali2020information}. 

Several mathematical notions have been introduced, mostly with respect to the relative importance of nodes in the underlying network, to formulate the concept of fairness~\cite{dwork2012fairness, Feldman2014CertifyingAR,jalali2020information}. One which has found significant popularity recently is the PageRank fairness~\cite{tsioutsiouliklis2021fairness,tsioutsiouliklis2022link,prfairwww24}. PageRank (PR) assigns a score to each node $u$ that signifies the importance of $u$ in the network globally, whereas Personalized PageRank (PPR) rooted at a specific node $v$ assigns a score to each node $u$ that captures the relative importance of $u$ for $v$~\cite{Anatomy1998,PRbeyond,PRwww23,Chung2010PageRankAR}.

We focus on group-based fairness, where nodes belong to groups based on the value of some protected attribute. For example, in a social or cooperation network where nodes correspond to individuals, the protected attribute may be age, gender, or religion. Previous research has shown that under certain conditions, the results of both PR and PPR may be unfair in terms of the PageRank weights assigned to each group~\cite{EspinNoboa2021InequalityAI,tsioutsiouliklis2021fairness,prfairwww24}. This is crucial since these measures not only have strong correlation to information access level for the nodes, but also are commonly used by many network algorithms~\cite{jeh2003scaling,Yin2019ScalableGE,Chen2020ScalableGN}, which could cause further discrimination against such disadvantaged groups.

\begin{figure}[t!]
    \centering
    \includegraphics[width=0.6\columnwidth]{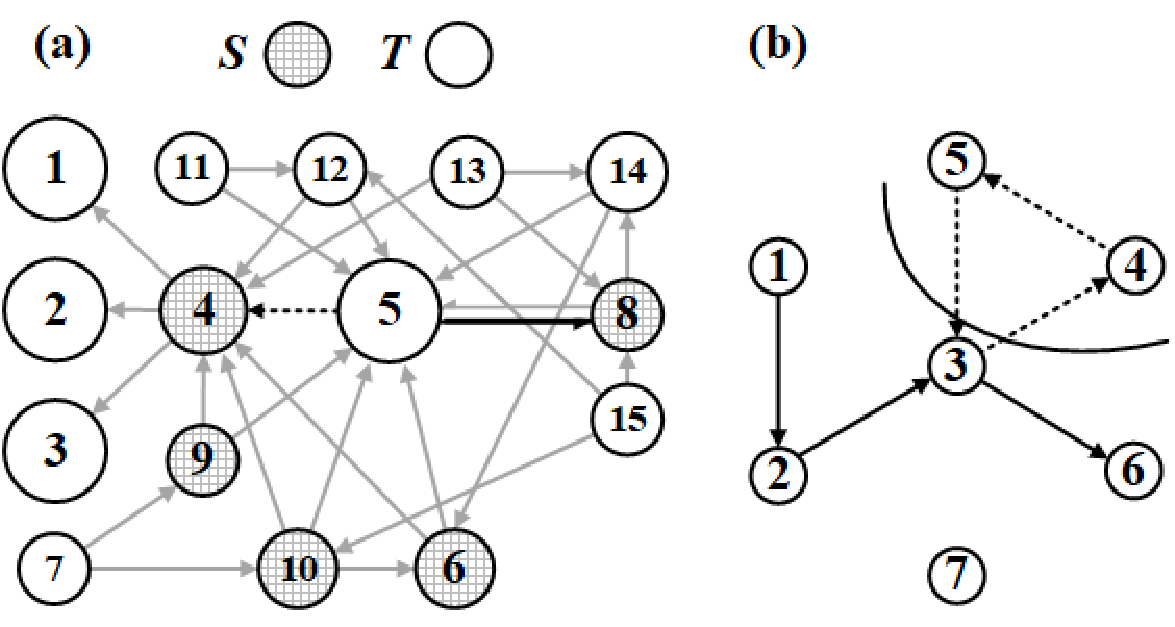}
    \caption{(a) Illustration of our proposed algorithm \textsc{Exact} (Algorithm 1) which rewires edge $(5,4)$ to $(5,8)$ and the PageRank allocation to the disadvantaged group $S$ changes from $\mpi(S)=0.13$ to $\mpi(S)=0.17$. (The sizes of nodes are determined by their PageRank scores.) (b) Illustration of the loop erase process of Wilson algorithm~\cite{Wilson1996GeneratingRS}, which bases our fast algorithm. The loop $(3,4,5,3)$ is erased.}
    \label{fig:illu}
\end{figure}

There has been a growing demand for a deeper understanding of how the unfairness in social networks can be mitigated by modifying the graph structure. In particular, various graph operations such as edge addition, deletion and rewiring (motivated by real-world instances such as link recommendation systems) have been explored~\cite{fairdrop21,li2021on,burst_2020,jalali2020information,swift2022maximizing,bashardoust2022reducing,tsioutsiouliklis2022link,LiZhAh24}. We focus on the popular edge rewiring operation that replaces an edge $(i,j)$ with a missing edge $(i,k)$ (see Figure~\ref{fig:illu} (a) for illustration.), which has been practically widely adopted across domains such as opinion dynamics~\cite{chitra2020analyzing,bhalla2023local,santos2021link}, recommendation systems~\cite{Fabbri2022RewiringWR,Coupette2023ReducingET}, graph learning~\cite{ma2021graph,diffwire22} and others~\cite{valente2012network,d2019rewiring,kim2012network,han2019measuring,Chan2016OptimizingNR}. One desired characteristic of the rewiring operation is that it preserves some essential graph properties, namely number of edges and average degree, in contrast to the edge insertion/deletion~\cite{ma2021graph}. Furthermore, the connectivity of the graph, formulated in a Laplacian based algebraic notion, alters less aggressively~\cite{ma2021graph}.

Given a graph $\calG=(V,E)$, a disadvantaged group $S\subset V$ and a budget $b$, we are interested in the problem of rewiring $b$ edges to maximize the total PR (analogously PPR) for nodes in $S$ (which results in minimizing PR (resp. PPR) unfairness towards set $S$). It is worth emphasizing that the aim here is to make the network itself fairer, which is different from in-processing approaches where the input of the algorithm is augmented for the duration of the algorithm~\cite{rastegarpanah2019fighting,wang2019enhancing} to reach temporary fairness. Our work can be seen as a continuation of works in~\cite{tsioutsiouliklis2021fairness,tsioutsiouliklis2022link}, which focused on adding/removing edges adjacent to a fixed node. (Please see Section~\ref{sec:related-work} for more details.)

In summary, we make the following contributions.
\begin{itemize}[leftmargin=*]
    \item We formally introduce the optimization problems to choose a set of edge rewirings that maximize the PR and PPR fairness.
    \item We provide analytical formulas for the effect of a single edge rewiring on PR and PPR fairness. Based on the formulas, we propose cubic time algorithms that exploit greedy strategies to solve the studied problems.
    \item To speed up the computation, we then first present a novel forest interpretation for the PageRank algorithm and subsequently propose efficient linear-time edge rewiring algorithms for maximizing PR and PPR fairness that exploits fast sampling of spanning forests.
    \item We report experimental results on real-world graphs which confirm the superiority of the proposed algorithms to the existing algorithms. For example, our linear-time algorithm could handle the LinkedIn network, which has more than 3 million nodes in just a few minutes.
\end{itemize}

\textbf{Roadmap.} 
We first introduce some necessary preliminaries related to our work in Section~\ref{sec:pre}. Then, in Section~\ref{sec:problem} we define the PageRank fairness and formulate the problems we aim to solve. In Section~\ref{sec:exact}, we present the first algorithm \textsc{Exact} that greedily selects the edges to be rewired, followed by Section~\ref{sec:fast} where we provide the algorithm \textsc{Fast} to speed up the greedy process. In Section~\ref{sec:exp}, we report our experimental findings. In Section~\ref{sec:related-work}, we discuss the related work. Finally, in Section~\ref{sec:conclusion}, we conclude this paper and discuss some future directions.

\section{Preliminary}\label{sec:pre}
This section sets up the stage for our study by introducing basic notations, the spanning (rooted) forest, forest matrix, and the PageRank algorithm.

\subsection{Notations}
In this paper, unless otherwise specified, we use boldface lowercase $\xx$ and uppercase $\MM$ letters to denote vectors and matrices, respectively. Individual elements are represented as $\xx_i$ for vectors (or $\MM_{ij}$ for matrices). We use $\MM_{i,:}$ and $\MM_{:,j}$ to refer to matrix rows and columns, respectively. $\ee_i$ represents the $i$-th standard basis vector with 1 at the $i$-th position and 0 elsewhere. $\one$ signifies an all-ones vector. The transpose of $\xx$ is denoted as $\xx^{\top}$, and $\one_S$ denotes a column vector with 1s at elements in set $S$ and 0s elsewhere.

Consider a directed weighted graph (digraph) $\calG = (V, E, w)$ where $V$ is the set of nodes, $E \subseteq V \times V$ is the set of directed edges and $w:E\rightarrow \mathbb{R}^{+}$ is the edge weight function. Let $n := |V|$ and $m := |E|$ denote the number of nodes and the number of edges, respectively. We use $\mathcal{N}\left(i\right)$ to denote the set of out-neighbors of $i$. The Laplacian matrix of $\calG$ is $\LL = \DD - \A$, where $\A = (a_{ij})_{n\times n}$ is the adjacency matrix whose entry $\A_{ij}=w(i,j)$ if $(i,j)\in E$, and $\A_{ij}=0$ otherwise, and $\DD$ is the out-degree diagonal matrix $\DD=\text{diag}(\dd_1,\cdots,\dd_n)$ where $\dd_i=\sum_j\A_{ij}$ is the out-degree of node $i$. We denote the maximum out-degree in $\calG$ as $d_{max}=\max\{\dd_i\mid i\in V\}$. Let $\PP=\DD^{-1} \A$ be the random walk matrix (i.e., transition matrix) of $\mathcal{G}$, in which $p_{i j}=\frac{w(i,j)}{\dd_{i}}$ if $e_{ij} \in E$ and $p_{i j}=0$ otherwise. We summarize all the frequently used notations in Table~\ref{tab:notations}.

{
\begin{table}
\centering
\caption{Frequently used notations.}
\begin{tabular}{|m{2.4cm}|m{8cm}|}
\hline
\textbf{Notation} & \textbf{Description} \\
\hline
$\calG=(V, E,w)$ & A directed graph $\calG$ with node set $V$, edge set $E$, and weight function $w$. \\
\hline
$n, m$ & The numbers of vertices and edges of $\calG$ respectively.\\
\hline
$\A, \DD, \PP, \LL$ & The adjacency, degree, transition and laplacian matrices of $\calG$. \\
\hline
$\bOmega, \omega_{i,j}$ & The forest matrix and its element at $i$-th row and $j$-th column. \\
\hline
$\mathcal{H}, \varepsilon(\mathcal{H})$ & A subgraph of $\calG$ with its weight being $\varepsilon(\mathcal{H})$. \\
\hline
$\calF, \calF_{ij}$ & The set of all spanning rooted forests, the subset of $\calF$ where node $j$ is rooted at node $i$ respectively. \\
\hline
$F, \rho(F)$ & An instance of $\calF$ and its root set.\\
\hline
$ \Delta((i,j,k), S)$ & The increase of PageRank (PR) allocation for group $S$ after the execution of the rewiring $(i,j,k)$.\\
\hline
$\Delta_v((i,j,k), S)$ & The increase of the allocation of Personalized PageRank (PPR) for group $S$ after the execution of the rewiring $(i,j,k)$.\\
\hline
\end{tabular}
\label{tab:notations}
\end{table}
}

\subsection{Spanning Forests and Forest Matrix}\label{sec:2.2}
For a digraph $\calG=(V, E, w)$, a spanning subgraph $\calH$ is a graph who has the same node set as $\calG$ and its edge set is a subset of $E$. A tree is a subgraph that has no cycles in it. A spanning forest on $\mathcal{G}$ is a spanning subgraph of $\mathcal{G}$ that is a forest, i.e., a set of disjoint trees. A spanning rooted tree is a spanning tree that has a node as the root. A spanning rooted forest of $\mathcal{G}$ is a spanning forest of $\mathcal{G}$, where all the trees are rooted. See Figure~\ref{fig:tree} as an illustration. Let $\calF$ denote the set of all spanning rooted forests of digraph $\mathcal{G}$ and $\calF_{i j}$ denote the set of those spanning forests of $\mathcal{G}$ that nodes $i$ and $j$ in the same tree rooted at node $i$. For a forest $F\in \calF$, we denote the set of its roots as $\rho(F)$, the set of nodes in it that rooted at node $i$ as $M(F, i)$, and the root of node $i$ in it as $r(F, i)$. We denote the weight of a forest $F$ as $\varepsilon(F)$ where $\varepsilon(F)$ is the product of the weights of all edges in $F$. We set the weight of a forest with no edges as 1. For any nonempty set $H$ of subgraphs, we define its weight as $\varepsilon(H)=\sum_{\mathcal{H} \in H} \varepsilon(\mathcal{H})$. If $H$ is empty, we set its weight to be zero~\cite{ChSh97,ChSh98}. 

Let $\II$ be the identity matrix, then the forest matrix~\cite{Chebotarev2006SpanningFA, Golender1981GraphPM} $\bOmega=\bOmega(\mathcal{G})$ of graph $\mathcal{G}$ is defined as $\bOmega=(\LL + \II)^{-1}=\left(\omega_{i j}\right)_{n \times n}$, where the entry $\omega_{i j}=\varepsilon\left(\calF_{j i}\right) / \varepsilon(\calF)$~\cite{ChSh97,ChSh98, Ch82}. For any pair of nodes $i$ and $j$ in graph $\mathcal{G}, 0 \leq \omega_{i j} \leq 1$~\cite{Merris1998DoublySG}. The forest matrix $\bOmega$ has various practical applications~\cite{Fouss2007RandomWalkCO,Jin2019ForestDC,Senelle2013TheSD, sun2023opinion}. For instance, its entry $\omega_{i j}$ can be used to gauge the proximity between nodes $i$ and $j$: the less the value of $\omega_{i j}$, the ``further'' $i$ is from $j$~\cite{ChSh97}. 

\begin{figure}[t!]
    \centering
    \includegraphics[width=0.8\columnwidth]{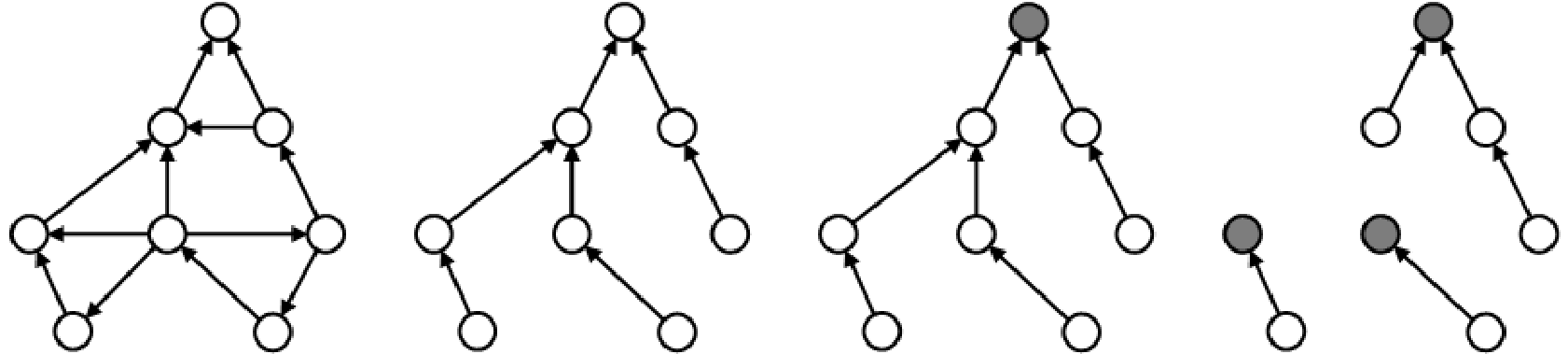}
    \caption{From left to right: a digraph, a spanning tree, a spanning rooted tree and a spanning rooted forest on it (roots are in gray).}
    \label{fig:tree}
\end{figure}

\subsection{The PageRank Algorithm}\label{sec:2.3}

The illustrious PageRank (PR) algorithm, initially unveiled in~\cite{Anatomy1998}, serves as a crucial component in link analysis and has found a broad-spectrum applicability in diverse domains, such as web page search, recommendation, disinformation detection, and social network analysis~\cite{PRbeyond,fu2022disco,fu2021sdg,kamvar2003extrapolation,Tong2006FastRW,yao2012anomaly,yoon2019fast}.

The standard expression of PageRank is expressed as follows
\begin{equation}\label{eq:preq1}
    \mpi^\top=(1-\alpha) \mpi^\top \PP+\alpha \vvv^\top,
\end{equation}
where $\mpi$ is a scoring vector for each node $v \in V$ which epitomizes the stationary distribution of a random walk with restart on graph $\calG$, $\PP$ is the transition matrix of this random walk, $\alpha$ is the parameter of restart probability, and $\vvv$ is the jump vector (also referred to as teleport vector or personalized vector~\cite{Chung2010PageRankAR, PRbeyond}), which determines a distribution over the nodes in the graph for restart node selection. Commonly, $\alpha=0.15$, and the jump vector is the uniform vector $\uu$~\cite{PRbeyond}. We denote the PageRank score of node $i$ as $\mpi(i)$.

Setting the restart vector as a unit vector $\ee_i$ yields an important extension of the PageRank algorithm: the Personalized PageRank (PPR) algorithm. The PPR vector for node $i$ is denoted as $\mpi_i$, signifying that node $i$ bequeaths PageRank $\mpi_i(u)$ to node $u$.

The following lemma will facilitate our analysis.
\begin{lemma}\label{lem:PII}
    For the PageRank vector $\mpi$, the following relationship holds: $\mpi^{\top}=\vvv^{\top} \mPi$, where: (\romannumeral1) $\mPi=\alpha(\II-(1-\alpha) \PP)^{-1}$; and (\romannumeral2) The $i$-th row vector $\mPi_{i,:}$ aligns with the personalized PageRank vector of node $i$: $\mpi_i=\mPi_{i,:}$.
\end{lemma}

\section{Problem Definition}\label{sec:problem}
\textbf{PageRank Fairness.}
Consider a node set $S$, which corresponds to a group of disadvantaged individuals. Let $r(S)=\frac{|S|}{|V|}$ denote the ratio of group $S$ in the overall population. We define $\mpi(S)$ to denote the PageRank mass allocated to $S$, i.e., $\mpi(S)=\sum_{i \in S} \mpi(i)$. We also refer to $\mpi(S)$ as the PR ratio of $S$.

For a group $S$ and a parameter $\phi$, we say that the network is $PR$-\textit{unfair} to group $S$ if $\mpi(S)<\phi$. While the parameter $\phi$ can be chosen to implement various notions of fairness, we set its default value to $\phi=r(S)$. In other words, the fraction of the PageRank weights allocated to $S$ is at least proportional to its size.

Analogously, for a node $v$, $\mpi_v(S)$ is used to denote the personalized PageRank mass assigned to group $S$ by node $v$, i.e., $\mpi_v(S)=$ $\sum_{i \in S} \mpi_v(i)$. As in~\cite{tsioutsiouliklis2021fairness}, for defining PPR-unfairness, the probability mass $\alpha$ allocated to $v$ is excluded through the random jump to consider only the ratio of the organic PPR that is assigned to $S$. More accurately, we define $\overline{\mpi_v}(S)=\frac{\mpi_v(S)-\alpha \one_{v \in S}}{1-\alpha}$, where $\one_{v \in S}$ is an indicator function which is 1 if and only if $v$ is in $S$. For a given parameter $\phi$, a node $v$ is said to be $PPR$-\textit{unfair} to $S$ if $\overline{\mpi_v}(S)<\phi$.

\textbf{Problem Statement.} We study the following optimization problems where the goal is to maximize fairness using edge rewiring denoted as $(i, j, k)$, in which we replace an edge $(i, j) \in E$ by an edge $(i, k) \notin E$ with $i \neq k$.

\begin{tcolorbox}[colback=white,colframe=black]
\begin{problem}\label{pro:1}
    (\uline{\textbf{P}}R F\uline{\textbf{a}}irness Max\uline{\textbf{i}}mization by Rewi\uline{\textbf{r}}ing, \pair) Given a digraph $\calG = (V,E,w)$, a target set $S\subset V$ and a budget $b$, we aim to find a collection $R$ of $b$ edge rewirings such that the PR allocated to set $S$, $\mpi(S)$ is maximized. 
\end{problem}
\end{tcolorbox}
\begin{tcolorbox}[colback=white,colframe=black]
\begin{problem}\label{pro:2}
    (\uline{\textbf{PP}}R F\uline{\textbf{a}}irness Max\uline{\textbf{i}}mization by Rewi\uline{\textbf{r}}ing, \ppair) Given a digraph $\calG = (V,E,w)$, a source node $v\in V$, a target set $S\subset V$, and a budget $b$, we aim to find a collection $R$ of $b$ edge rewirings  such that the PPR allocated to set $S$, $\overline{\mpi_v}(S)$ is maximized. 
\end{problem}
\end{tcolorbox}

Note that our goal is to assist a disadvantaged group $S$ for which $\mpi(S)< \phi=r(S)$ by increasing $\mpi(S)$. This should explain our choice of the objective function in Problem 1. (A similar argument applies to Problem 2.)

\section{Exact Greedy Algorithm}\label{sec:exact}
Obviously, the proposed problems, namely \textsc{Pair} and \textsc{PPair} are inherently combinatorial, imposing many challenges in solving them. Fortunately,  a greedy approach often yields commendable results for these kinds of problems. In this paper, we also resort to the greedy strategy. We first analyze the mechanics of individual rewiring to understand how we can identify and perform greedily optimal rewirings efficiently. As each greedy rewiring $(i,j,k)$ constitutes a rank-one update of the transition matrix $\PP$, we can express the updated transition matrix $\PP^{\prime}$ as
$$
\PP^{\prime}=\PP+\aaa(-\bb)^{\top},
$$
where $\aaa=p_{i j} \ee_i$ and $\bb=\ee_j-\ee_k$. This corresponds to a rank-one update of $\mPi$ defined in Lemma~\ref{lem:PII}, giving rise to the updated matrix:
$$
\mPi^{\prime}=\alpha(\mathbf{I}-(1-\alpha) (\PP+\aaa(-\bb)^{\top}))^{-1}.
$$
Leveraging the rank-1 update, the Sherman-Morrison formula~\cite{Me73} enables us to articulate the updated matrix $\mPi^{\prime}$ as 
\begin{align}
    \label{eq:sm}
    \mPi^{\prime} = \mPi - \frac{(1-\alpha)\mPi\aaa\bb^{\top}\mPi}{\alpha+(1-\alpha)\bb^{\top}\mPi\aaa}.
\end{align}
Let $\sigma = \frac{1}{n}\one^\top\mPi, \eta=\mPi\one_S$, and $\tau=\alpha+(1-\alpha)\bb^{\top}\mPi\aaa$. Therefore, the PR fairness gain for group $S$ of rewiring edge $(i, j)$ to $(i,k)$ is 

 \begin{align}
    \Delta((i,j,k), S) = \vvv^{\top}\mPi^\prime\one_{S} - \vvv^{\top}\mPi\one_{S} 
      = -\vvv^\top\frac{(1-\alpha)\mPi\aaa\bb^{\top}\mPi\one_{S}}{\alpha+(1-\alpha)\bb^{\top}\mPi\aaa}
     = -(1-\alpha)\cdot\frac{\sigma\aaa\bb^\top\eta}{\tau}.\label{eq:prgain}
 \end{align}

Let $\tilde{\sigma} = \ee_v^{\top}\mPi$ which is exactly the $v$-th row of matrix $\mPi$. The PPR fairness gain for group $S$ with regard to a source node $v$ is 
 \begin{align}
     \Delta_v((i,j,k), S) = \ee_v^{\top}\mPi^\prime\one_{S} - \ee_v^{\top}\mPi\one_{S}
      = -\frac{(1-\alpha)\ee_v^\top\mPi\aaa\bb^{\top}\mPi\one_S}{\alpha+(1-\alpha)\bb^{\top}\mPi\aaa}
     = -(1-\alpha)\cdot\frac{\Tilde{\sigma}\aaa\bb^\top\eta}{\tau}\label{eq:pprgain}.
 \end{align}
 
The $i$-th element in $\sigma$ can be considered as a measure of centrality of node $i$. Specifically, it represents the average allocation of PPR to node $i$ for all the nodes. A higher value signifies node $i$'s increased centrality in the network. The $i$-th element in $\tilde{\sigma}$ is exactly the PPR value that node $v$ allocates to node $i$. The $i$-th element in $\eta$ can be considered as a measure of proximity between node $i$ and group $S$. It sums up the PPR allocated by node $i$ to the nodes in set $S$. A larger value implies node $i$'s greater proximity to the group $S$.

\begin{lemma}
    $\tau$ is always positive. $0\leq\sigma_i\leq 1$, $0\leq\tilde{\sigma}_i\leq 1$, and $0\leq\eta_i\leq 1$ for all $i\in V$.
\end{lemma}
\begin{proof}
    For a node $x$, $\frac{1-\alpha}{\alpha}p_{ij}\mPi_{xi}$ is the probability we traverse the edge $(i,j)$ in a walk starting at $x$. Now, the probability that we reach $j$ from $k\notin\{i,j\}$ is at most $(1-\alpha)$, and the probability that we traverse $(i,j)$ from $k$ without first visiting $j$ is at most $(1-\alpha)^2p_{ij}$. Since $\alpha>0$ and $p_{ij}\leq 1-\alpha$, therefore, we have
    \begin{align}
        \notag\frac{1-\alpha}{\alpha}p_{ij}\mPi_{ki}\leq(1-\alpha)^2p_{ij}+\frac{(1-\alpha)^3}{\alpha}p_{ij}\mPi_{ji}
    \leq 1+\frac{1-\alpha}{\alpha}p_{ij}\mPi_{ji},
    \end{align}
    and hence, we have $\tau\geq0$.
    
    From the definitions of $\sigma$, $\tilde{\sigma}$ and $\eta$, we can naturally derive that every element in them is in $[0,1]$.
\end{proof}

\textit{Intuitions.} $\tau$ is driven by the difference between two entries of $\mPi$, whereas $\sigma\aaa$ arises from the average of a single column of  $\mPi$ and $\bb^\top\eta$ is driven by the difference between two partial row sums (indexed by $S$) of $\mPi$. Thus, the variation in $\sigma\aaa\bb^\top \eta$ generally surpasses that in $\tau$, and large $\sigma \aaa\bb^\top\eta$ values mostly dominate small values of $\tau$. Therefore, intuitively, we argue that we just need to select a node $i$ that is relatively central in the network (measured by $\sigma_i$), disconnect it from a neighbor $j$ that is relatively distant from the set $S$ (measured by $\eta_j$), and instead connect it to a neighbor $k$ that is relatively close to the set $S$ (measured by $\eta_k$). This intuition is verified by the example network in Figure~\ref{fig:illu} (a), where the rewiring $(5,4,8)$ maximizes PR fairness of the disadvantaged group $S$. Node 5 is obviously central in the network and $\mpi_4(S) = 0.1< \mpi_8(S)=0.21$. This argument is the main intuition behind our strategy of reducing candidates, explained in the next section. As for the optimization of PPR fairness, the variation in $\tilde{\sigma}\aaa\bb^\top \eta$ generally surpasses that in $\tau$, and large $\tilde{\sigma}\aaa\bb^\top\eta$ values mostly dominate small values of $\tau$. Therefore, intuitively, we argue that we just need to select a node $i$ that is relatively close to the source node $v$ in the network (measured by $\tilde{\sigma}_i$), disconnect it from a neighbor $j$ with the smallest value in $\eta$, and instead connect it to a neighbor $k$ with the 
largest value in $\eta$. 

We now introduce our first algorithm, designated as \textsc{Exact}, the pseudocode of which is delineated in Algorithm~\ref{al-optimal}. The \textsc{Exact} algorithm employs an edge rewiring strategy to systematically modify the graph structure by iteratively selecting the most advantageous local edge rewiring at each step. The algorithm commences by initializing an empty solution set $R$. Subsequently, it proceeds to augment $R$ with $b$ edge rewirings through an iterative process. During each iteration, \textsc{Exact} meticulously examines all edges $(i, j) \in E$ within the graph. For each edge, it computes the impact $\Delta((i,j,k), S)$ of rewiring to every potential target node $k \in V$. This comprehensive evaluation ensures that the algorithm identifies the optimal rewiring choice at each step. The computational complexity of \textsc{Exact} is primarily determined by three key operations: the calculation of the matrix $\mPi$, which, using a naive approach, requires $O(n^3)$ time; the computation of $\Delta$, which, while taking constant time $O(1)$ for each calculation, is executed $O(mn)$ times throughout the algorithm's run; and the updating of $\mPi$ after each rewiring, necessitating $O(n^2)$ time per update. Aggregating these computational requirements and considering that the process is repeated $b$ times for the $b$ rewirings, we arrive at an overall time complexity of $O(n^3 + bmn)$ for the \textsc{Exact} algorithm. 

\begin{remark}\label{rem:exactv}
    The algorithm for exactly optimizing PPR fairness is similar to \textsc{Exact} which could be obtained by replacing $\Delta((i,j,k), S)$ in Line 4 of Algorithm~\ref{al-optimal} with the PPR gain $\Delta_v((i,j,k), S)$ defined in Eq.~\eqref{eq:pprgain}. We call the exact greedy algorithm for optimizing PPR fairness \texttt{Exactv}.
\end{remark}

\normalem
\begin{small}
\begin{algorithm}[tb]
	\caption{\textsc{Exact}$(\calG, b, S)$}
	\label{al-optimal}
	\Input{
		$\calG$: a digraph; $b$ : the rewiring budget; $S$: the target group
	}
	\Output{
		$R$: A series of edge rewirings with $|R|=b$
	}
	Initialize solution $R= \emptyset$ \;
	Compute $\mPi$ defined in Eq.~\eqref{eq:22}\;
	\For{$ t=1 $ to $ b $}{
        $ (i^*, j^*, k^*) \leftarrow \argmax_{(i,j,k), s.t. (i,j)\in E, (i,k)\notin E, i\neq k} \Delta((i,j,k), S) $\;
		Update solution $R \gets R\cup \{ (i^*,j^*,k^*) \}$ \;	
        Update the matrix $\mPi$ according to Eq.~\eqref{eq:sm}\;
        Update graph $\calG \gets (V, E\setminus\{(i^*,j^*)\}\cup\{(i^*,k^*)\})$ \;	
	}	
	\Return $R$.
\end{algorithm}
\end{small}

\section{Fast Sampling Algorithm}\label{sec:fast}
In this section, we first establish a connection between our objective function and spanning rooted forests. We then develop a linear time algorithm to compute $\Delta$ and $\Delta_v$, which allows us to speed up the algorithm \textsc{Exact}. Our fast algorithm is based on the sampling of spanning rooted forests, the main ingredient of which is an extension of  Wilson’s algorithm~\cite{Wi96,WiPr96}. 

\subsection{Novel Interpretation for PageRank via Spanning Rooted Forest } 

In this section, we establish a connection between the PageRank derived matrix $\mPi$ (see Section~\ref{sec:2.2}) and the forest matrix (see Section~\ref{sec:2.3}) of a new reweighted graph. 

We first note that Eq.~\eqref{eq:preq1} can be transformed into an equivalent linear system as illustrated below:
\begin{align}\label{eq:preq2}
    \tilde{\mpi}(\LL+\bm{\Lambda})=\alpha/(1 - \alpha) \vvv,
\end{align}
where $\bm{\Lambda} = \alpha/(1 - \alpha) \DD $ and $\tilde{\mpi} = \mpi\DD^{-1}$. Drawing insights from the above equation, it can be deduced that:
\begin{align}\label{eq:22}
    \mPi=(\bm{\Lambda}^{-1}\LL + \II)^{-1},
\end{align}
which is exactly the forest matrix~\cite{Golender1981GraphPM, Chebotarev2006SpanningFA} of a new graph. We summarize the process of building this new graph in the following lemma.
\begin{lemma}\label{lem:reweight}
    Let $\calG_r=(V_r, E_r, w_r)$ denote the new graph, which is initially set empty. Then we perform the following steps: (\romannumeral1) Copy all the nodes in $V$ and all the edges in $E$ to $V_r$ and $E_r$ respectively; (\romannumeral2) For each edge $(i,j)\in E_r$, we set its weight as $w_r(i,j)=\frac{w(i,j)}{\alpha/(1 - \alpha)\dd_i}$.
\end{lemma}

The graph $\calG_r$ is actually a reweighted graph over graph $\calG$. After this transformation, we can conclude that computing the PageRank derived matrix $\mPi$ for graph $\calG$ is equivalent to computing the forest matrix for the reweighted graph $\calG_r$, which provides us the insight to design fast approximate algorithms to speed up the computations through Monte Carlo sampling. 

We now interpret the concerned quantities in $\Delta$ and $\Delta_v$ using spanning rooted forests. Let $\bOmega(\calG_r)$ denote the forest matrix of the reweighted graph $\calG_r$ and $\omega_{ij}$ denote its element in the $i$-th row and $j$-th column. It is obvious that $\bOmega(\calG_r) = \mPi$. We now provide a novel interpretation and another expression of $\sigma_s$, $\tilde{\sigma}_s$ and $\eta_s$ for each node $s\in V_r$. Recall that $\sigma_s=\frac{1}{n}\sum_{u\in V}\mpi_u(s)$ $\tilde{\sigma}_s = \mpi_u(s)$, and  $\eta_s=\sum_{u\in S}\mpi_s(u)$. Together with the fact that $\omega_{i j}=\varepsilon\left(\calF_{j i}\right) / \varepsilon(\calF)$~\cite{ChSh97,ChSh98, Ch82}, where $\calF$ and $\calF_{ji}$ are  sets of spanning rooted forests of graph $\calG_r$ (see Section 2.2), we have the following theorem.

\begin{theorem}\label{the:interpretation}
     For a graph $\calG = (V, E)$, and any node $s\in V$, we can interpret the concerned quantities as 
    \begin{gather} 
        \notag\sigma_s = \frac{1}{n}\sum_{u \in V}\omega_{us}  = \frac{\sum_{F\in\calF}|M(F, s)|\varepsilon(F)}{n\varepsilon(\calF)},\, 
        \notag\Tilde{\sigma}_s = \omega_{vs} = \frac{\varepsilon(\calF_{sv})}{\varepsilon(\calF)},\\
        \notag \eta_s = \sum_{u \in S}\omega_{su} = \frac{\sum_{u \in S}\varepsilon(\calF_{us})}{\varepsilon(\calF)}.
    \end{gather}
\end{theorem}

The above theorem indicates that $\sigma_s$ is the sum of the weight $\varepsilon(F)$ of all $F\in \calF$ weighted by the number of nodes in the component of $F$ that rooted at $s$, divided by $n$ and the total weight of all spanning rooted forests, i.e., $\varepsilon(\calF)$. $\Tilde{\sigma}_s$ is the weight of $\calF_{sv}$ divided by the weight of $\calF$. $\eta_s$ is the sum of the weight of all spanning forests in which $s$ is rooted at a node in $S$, divided by $\calF$'s weight. 

\textbf{Reducing the Number of Candidates.} The reader may have noticed that in this section, we provide fresh interpretations for all the components except $\tau$ in $\Delta$ and $\Delta_v$. This is because, as highlighted in the \textit{intuitions} in Section~\ref{sec:exact}, $\tau$ could be safely omitted without any significant impact. Additionally, we note that there is no need to consider all possible rewiring targets, but instead focus on a limited set $K$ of top nodes in $\eta$, which could be identified in $O(n)$ time. We set the restricted rewiring target node set $K$ as the top $d_{max}$ nodes that have the largest value in $\eta$, which ascertains that despite restricting the targets, for each edge $(i,j)\in E$, there is at least one unconnected node in $K$. Hereafter, instead of using $\Delta$, we opt to compute $\widehat{\Delta} = \Delta\tau$ and focus on a limited set $K$ of top nodes in $\eta$ as a heuristic. In the same manner, we use $\widehat{\Delta}_v = \Delta_v\tau$ to replace $\Delta_v$, this strategy is widely used in prior works, such as in~\cite{Coupette2023ReducingET}. 

\subsection{Approximating via an Extension of Wilson's Algorithm}
According to Theorem~\ref{the:interpretation}, we can estimate $\sigma$, $\Tilde{\sigma}$ and $\eta$ by sampling spanning rooted forests in the reweighted graph $\calG_r$. Specifically, if we can sample spanning forests according to their weights that is $\mathbb{P}\left(\calF=F\right) \propto \varepsilon(F)=\prod_{(i, j) \in F} w_r(i, j)$, then an unbiased estimator of $\sigma$, $\Tilde{\sigma}$ and $\eta$ can be easily derived. The challenge lies in sampling from this weighted distribution. To address this, we propose a loop-erased random walk approach to sample a spanning rooted forest according to its weight. This method extends the classic Wilson algorithm for sampling spanning trees on graphs~\cite{Wilson1996GeneratingRS}.

We first briefly introduce the loop-erasure operation on a random walk~\cite{La80}, a process involving the removal of loops in the order they occur. In a graph $\calG$ and a random walk trajectory $\gamma=\left(v_1, \cdots, v_l\right)$ on $\calG$, we define the loop-erased trajectory $\operatorname{LE}(\gamma)=\left(v_{i_1}, \cdots, v_{i_j}\right)$ by eliminating loops. Specifically, $i_j$ is determined iteratively: $i_1=1$, and $i_{j+1}=\max \left\{i \mid v_i=v_{i_j}\right\}+1$. Loop-erased random walks avoid self-intersections and halt when revisiting previous trajectory. For example, in Figure~\ref{fig:illu} (b), for a random walk $\gamma=\left(1, 2, 3, 4, 5, 3, 6\right)$, its loop-erased trajectory becomes $\operatorname{LE}(\gamma)=\left(1, 2, 3, 6\right)$ after removing the loop $\left(3, 4, 5, 3\right)$. To apply loop erasure efficiently, we record the next node during the random walk, and upon completion, we retrace the trajectory by starting from the beginning and progressing to the recorded next node until encountering the prior trajectory.

\normalem
\begin{small}
\begin{algorithm}
	\caption{$\textsc{RandomForest}(\calG)$}
	\label{alg:rf}
	\Input{ $\calG=(V,E, w)$}
	\Output{$Root[u] \forall u\in V$}
	$InForest[u]  \leftarrow  false, Next[u] \leftarrow -1, Root[u]\leftarrow 0 \forall u\in V$\;
	\For{$ i = 1 $ to $ n $}
	{$ u \leftarrow i $\; 
		\While{not $InForest[u]$}{
			\If{\textsc{Rand}()$\leq \frac{1}{1 + \dd_u}$}{
				$InForest[u]  \leftarrow  true$\;
                $Root[u] \leftarrow u$\; 
                $Next[u] = -1 $\;
			}
			\Else{
				$Next[u] \leftarrow  \textsc{RandomNeighbor}(\calG, u$)\; 
				$u \leftarrow Next[ u ]$\;
			}
		}
	  $r \leftarrow  	Root[u ],  u\leftarrow i $\;
		\While{not InForest[$ u $]}{
			$InForest[u] \leftarrow true$\;
            $Root[u ] \leftarrow r$\;
            $u \leftarrow  Next[ u ]$\;
		}
	}
 \Return $Root$ \;
\end{algorithm}
\end{small}

Wilson's algorithm, based on loop-erasure applied to a random walk, was introduced to generate a uniform spanning tree rooted at a specified node~\cite{Wi96, WiPr96}. The algorithm consists of three key steps: initializing a tree with the root node, performing an unbiased random walk until it reaches a node within the tree, and applying the loop-erasure operation to add nodes and edges to the tree. This process continues until all nodes are included in the tree, ensuring an equal probability distribution for all possible spanning trees~\cite{Wi96, WiPr96}.

For a digraph $\calG= (V,E, w)$, we can adapt Wilson's Algorithm to produce a spanning rooted forest $F \in \calF$, following a procedure akin to that in~\cite{AvLuGaAl18,PiAmBaTr21}. This process involves three essential steps:
\begin{itemize}[leftmargin=*]
    \item The construction of an extended digraph $\calG'=(V',E', w^\prime) $ from $\calG= (V,E, w)$ by introducing a new node $\xi$ and the requisite edges. Specifically, $V' =V \cup \{\xi\}$ and $E'= E\cup \{(i,\xi), (\xi,i)\}$ with $w^\prime(i,\xi) = w^\prime(\xi, i) = 1$ for all $i\in V$. The weights of the original edges in $E$ remain unchanged in $E'$.
    \item The application of Wilson's algorithm to generate a uniform spanning tree $T$ for graph $\calG'$ with $\xi$ as the root node.
    \item The removal of edges $(i,\xi) \in T$ to obtain a spanning forest $F \in \calF$ of $\calG$. By designating $\rho(F) = \{ i\mid(i,\xi)\in T \}$ as the set of roots for trees in $F$, we transform $F$ into a spanning rooted forest of $\calG$.
\end{itemize}

Following the three steps above, in Algorithm~\ref{alg:rf} we present an algorithm \textsc{RandomForest} to generate a uniform spanning rooted forest $F$ of digraph $\calG$. We track node membership in the forest using the bool vector $InForest$ and initialize $Next[u]$ which records the next node in random walk step to $-1$ for all nodes $u\in V$. We use $Root$ to record the root of every node. The algorithm initiates a random walk at node $u$ in the extended graph $\calG'$ (Line 3) to form a forest branch. The probability of reaching node $\xi$ from $u$ is $1/(1+d_u)$. If a random number generated in Line 5 (using \textsc{Rand}()) satisfies the given condition, the walk transitions to node $\xi$. As per our analysis, nodes directly connected to $\xi$ in $\calG'$ are part of the root set $\rho(F)$. When the condition in Line 5 is not met, \textsc{RandomNeighbor} is used to select a random node $i$ from $\calN(u)$ with a probability of $\frac{w(u,i)}{\sum_{j\in\calN(u)}w(u,j)}$ (Line 8). Subsequently, $u$ is updated to $Next[u]$, and the process returns to Line 4. The \textit{for} loop concludes when the random walk encounters a node already present in the forest. Afterward, a new branch is formed. In Lines 10-12, the branch's loop-erasure is added to the forest, along with an update to $Root$. In Line 13, the forest's weight is computed by multiplying the weights of all its edges. Finally, \textsc{RandomForest} returns $Root$  along with the weight of the sampled forest.

\begin{theorem}\label{the:complexity}
    For graph $\calG_r$, \textsc{RandomForest} runs in time $O(n)$.
\end{theorem}
\begin{proof}
    According to the results in~\cite{Wi96}, the expected run time of \textsc{RandomForest} is the mean commute time. In proposition of~\cite{Ma00}, the authors have written this commute time as matrix forms. Adapting his results, we have that in our setting the expected running time of \textsc{RandomForest} equals the trace of $(\bm{\Lambda}^{-1}\LL + \II)^{-1}(\II + (1-\alpha)/\alpha\II)$. We know that the diagonal element $\omega_{ii}$ in matrix $(\bm{\Lambda}^{-1}\LL + \II)^{-1}$ is in $(0, 1)$ for each $i\in V$. Therefore, we can further derive the bound of the expected running time that $\sum_{i=1}^n\omega_{ii}(1+(1-\alpha)/\alpha)\leq\sum_{i=1}^n(1+(1-\alpha)/\alpha)=n(1+(1-\alpha)/\alpha)$, concluding the proof.
\end{proof}

To approximate $\sigma$, $\tilde{\sigma}$ and $\eta$, we first construct a reweighted graph $\calG_r$ following the steps in Lemma~\ref{lem:reweight}. Then, using \textsc{RandomForest}, we generate a set $\widehat{\calF}$ of $\psi$ random spanning rooted forests $F_1,F_2,\ldots,F_\psi$ on $\calG_r$. Define 
\begin{gather}
    \notag\sigma_s^\prime = \frac{\sum_{j=1}^\psi|M(F_j,s)|}{n\psi},\,
    \notag\tilde{\sigma}_s^\prime = \frac{\sum_{j=1}^\psi\one_{r(F_j, v) = s}}{\psi}, \,
    \notag\eta_s^\prime = \frac{\sum_{j=1}^\psi\one_{r(F_j, s)\in S}}{\psi},
\end{gather}
where $\one_{r(F_j, s)\in S}$ and $\one_{r(F_j, v) = s}$ are indicator functions that equal 1 only when the condition at the subscript is met and otherwise equal 0. Based on the above formulations, we can define three random variables for a random spanning forest: $\Bar{\sigma}_s(F_i) = \frac{|M(F_i, s)|}{n}$, $\Bar{\Tilde{\sigma}}_{s}(F_i)=\one_{r(F_i, v)=s}$, $\Bar{\eta}_s(F_i)=\one_{r(F_i, s)\in S}$ for estimating $\sigma_s$, $\tilde{\sigma}_s$ and $\eta_s$ respectively. According to~\cite{AvLuGaAl18}, the spanning rooted forest $F$ obtained from \textsc{RandomForest} is selected from $\calF$ with a probability proportional to its weight. Thus, we have $\mathbb{E} \left(\Bar{\sigma}_s(F_i)\right)= \sigma_s$, $\mathbb{E} \left(\Bar{\Tilde{\sigma}}_{s}(F_i)\right)= \tilde{\sigma}_s$ and $\mathbb{E} \left(\Bar{\eta}_s(F_i)\right)= \eta_s$, which implies that the three variables proposed above are unbiased estimators of $\sigma_s$, $\tilde{\sigma}_s$ and $\eta_s$ respectively for all $s\in V$.

Generally, the larger the value of $\psi$, the more accurate the estimations are. Next, we bound the number $\psi$ of required samplings of spanning rooted forests to guarantee a desired estimation precision by applying the following Hoeffding's inequality.

\begin{theorem}
(\textsc{Hoeffding's Inequality}~\cite{Ho63}). Let $Z_1, Z_2, \ldots, Z_{n_z}$ be independent random variables, where $Z_i\left(\forall 1 \leq i \leq n_z\right)$ is strictly bounded by the interval $\left[a_i, b_i\right]$. We define the empirical mean of these variables by $Z=\frac{1}{n_z} \sum_{i=1}^{n_z} Z_i$. Then, we have
$$
\mathbb{P}[|Z-\mathbb{E}[Z]| \geq \varepsilon] \leq 2 \exp \bigg(-\frac{2 n_z^2 \varepsilon^2}{\sum_{i=1}^{n_z}\left(b_i-a_i\right)^2}\bigg)
$$
\end{theorem}

Then it is easy to derive the following theorem.
\begin{theorem}\label{th-Fast}
	For any error parameter $ \epsilon>0 $ and $ \delta\in(0,1) $, if the number of samples $\psi$ is chosen obeying $ \psi =\left \lceil \frac{1}{2\epsilon^2}\ln{\frac{2}{\delta} } \right \rceil  $, then for any $s \in V$, let $Z_1 = \frac{1}{\psi}\sum_{i=1}^{\psi}\Bar{\sigma}_s(F_i)$, $Z_2 = \frac{1}{\psi}\sum_{i=1}^{\psi} \Bar{\Tilde{\sigma}}_{s}(F_i)$, and $Z_3 = \frac{1}{\psi}\sum_{i=1}^{\psi} \Bar{\eta}_s(F_i)$, then we have $ \mathbb{P}\left\{ |Z_1- \sigma_s| >\epsilon \right\} \leq \delta $, $ \mathbb{P}\left\{ |Z_2- \tilde{\sigma}_s| >\epsilon \right\} \leq \delta $ and $ \mathbb{P}\left\{ |Z_3- \eta_s| >\epsilon \right\} \leq \delta $.
\end{theorem}

\normalem
\begin{small}
\begin{algorithm}[H]	\caption{$\textsc{Fast}\left(\calG, b, S, \psi\right)$}
	\label{alg:fast}
	\Input{$\calG$ : a digraph, $b$: the rewiring budget, $S$: the target group, $\psi$: the number of spanning rooted forests to sample}
	\Output{$R$ : the solution set }
	Initialize: $ R \leftarrow \emptyset $\; 
 \For{$i = 1$ to $b$}{
    Construct the reweighted graph $\calG_r$ based on $\calG$ following the steps in Lemma~\ref{lem:reweight}\;
    $\sigma_v^\prime\leftarrow 0, \eta_v^\prime \leftarrow 0 \forall v \in V$, $w_{sum} = 0$\;
    \For{$t = 1$ to $\psi$}{
		$Root\leftarrow $ \textsc{RandomForest}($ \calG_r $)\;
        \For{$ j=1 $ to $ n $}{
			$ u \leftarrow  Root[j]$\;
			$\sigma_u^\prime \leftarrow \sigma_u^\prime + 1$\;
            \lIf{$u\in S$}{
            $\eta_j^\prime \leftarrow \eta_j^\prime + 1$}}
	}
	$\sigma^\prime \leftarrow \sigma^\prime/(n\psi)$; $\eta^\prime \leftarrow \eta^\prime/\psi$ \;
    $K\gets$ the top $d_{max}$ nodes of $\eta^\prime$\;
    $ (i^*, j^*, k^*) \leftarrow \argmax_{{(i,j,k), s.t. (i,j)\in E, (i,k)\notin E, k \in K}} \widehat{\Delta}^\prime((i,j,k), S) $\;
    Update the solution set $ R \leftarrow R \cup \{(i^*,j^*,k^*)\} $\;
    Update the graph $\calG = (V, E\setminus\{(i^*,j^*)\}\cup\{(i^*,k^*)\})$\;
	}
	\Return $ R $ \;
\end{algorithm}
\end{small}

\subsection{Fast Approximation Algorithm}

In this section, we present the full efficient sampling-based algorithm \textsc{Fast} to approximately solve the problem \textsc{Pair} in linear time. The main ingredient of the approximation algorithm \textsc{Fast} is the variation of Wilson's algorithm introduced in the preceding section. The details of algorithm \textsc{Fast} are described in Algorithm~\ref{alg:fast}. \textsc{Fast} takes a digraph $\calG$, a budget $b$, the target group $S$, and the number of samples $\psi$ as input. Then, \textsc{Fast} runs in $b$ rounds for iteratively select $b$ edge rewirings. In each round, it first constructs the reweighted graph $\calG_r$ following Lemma~\ref{lem:reweight}, followed by sampling $\psi$ spanning rooted forests by calling \textsc{RandomForest} for computing $\sigma^\prime$ and $\eta^\prime$. Subsequently, it goes through all $m$ edges and the set $K$ of $d_{max}$ rewiring targets that have top values in $\eta^\prime$ and selects an edge rewiring $(i^*,j^*,k^*)$ which maximizes the PR fairness gain 
\begin{align}
    \widehat{\Delta}^\prime ((i,j,k), S) = (1-\alpha)p_{ij}\sigma_i^\prime(\eta_k^\prime - \eta_j^\prime).
\end{align} 

The following theorem characterizes the performance of \textsc{Fast}.

\begin{theorem}\label{th-performance}	
	For any $b>0$, an error parameter $\eps$, and a failure probability $ \delta\in(0,1) $, let $\psi=\left \lceil \frac{1}{2\epsilon^2}\ln{\frac{2}{\delta} } \right \rceil$, \textsc{Fast} outputs a solution set $R$ in time $O(b(n\psi + md_{max}))$ by iteratively selecting $b$ edge rewirings such that in each iteration, the PR fairness gain is maximized.
\end{theorem}

\begin{remark}\label{rem:fastv}
    The fast algorithm for optimizing PPR fairness is similar to \textsc{Fast} which could be obtained by replacing the $\widehat{\Delta}^\prime((i,j,k), S)$ in Line 14 of Algorithm~\ref{alg:fast} by the corresponding PPR gain $\widehat{\Delta}^\prime_v((i,j,k), S) = (1-\alpha)p_{ij}\tilde{\sigma}_i^\prime(\eta_k^\prime - \eta_j^\prime)$. We call this fast greedy algorithm for optimizing PPR fairness \textsc{Fastv}.
\end{remark}

\section{Experiments}\label{sec:exp}
This section experimentally evaluates the efficiency and accuracy of our proposed algorithms for optimizing PR and PPR fairness, compared to existing methods and baselines.

\subsection{Experimental Setup}
\textbf{Datasets.} To evaluate the efficiency and accuracy of our proposed algorithms for optimizing fairness, we conduct experiments on 6 real-world datasets of various sizes, as described in Table~\ref{tab:data} which are collected from $\mathrm{SNAP}$~\cite{LeSo16}. We refer to the disadvantaged group as $S$ and set the advantaged group to be $T=V\setminus S$. The split comes from the datasets. Detail descriptions for these datasets are as follows. 
\begin{itemize}[leftmargin=*]
    \item  Books: A network of books about US politics, where edges between books represented co-purchasing. The disadvantaged group is lesser-known books.
    \item Blogs:  A network of political blogs, disadvantaged group is blogs with lesser visibility. 
    \item DBLP-pub: The network is built from DBLP author collaborations in selected data mining and database conferences (2011–2020). The disadvantaged group includes authors whose first publication was in 2016 or later.
    \item DBLP-Gender: The same network as DBLP-pub but with the disadvantaged group being  underrepresented gender.
    \item DBLP-Aminer: An author collaboration network by Aminer academic search system using publication data from DBLP. The disadvantaged group here includes authors with fewer collaborations.
    \item LinkedIn: A professional network, disadvantaged group includes profiles with fewer connections. 
\end{itemize}
\textbf{Implementation Details.} All experiments are conducted on a Linux machine with an Intel Xeon(R) Gold 6240@2.60GHz 32-core processor and 128GB of RAM. We implemented our algorithms in Julia. For reproducibility, the source code is available at~\url{https://github.com/PRfairness/prfairness}. We report the running time (measured in wall-clock time) and the actual error of each algorithm on all datasets. We exclude a method if it fails to report the result within one day. For each dataset, we perform $b=50$ edge rewirings.  

\textbf{Comparison.} We compare our PR fairness algorithms, \textsc{Exact} and \textsc{Fast}, and PPR fairness algorithms, \textsc{Exactv} and \textsc{Fastv}, with the Random (RND) strategy which performs $b$ random edge rewirings.  We also compare with the existing algorithms\footnote{https://github.com/ksemer/fairPRrec} in~\cite{tsioutsiouliklis2022link} by adapting them to a sequence of edge deletion and edge addition. For ease of exposition, we refer to them as MFREC (for optimizing PR fairness) and MPREC (for optimizing PPR fairness). We also compare with an existing algorithm named RePBubLik\footnote{https://github.
com/CriMenghini/RePBubLik} (RBL) from~\cite{haddadan2021repbublik} that aims to reduce structural bias by adding links. We adapted RBL to edge rewiring setting by first obtaining a set of potential new edges returned by RBL and then checking the existing edges to find suitable edge rewirings according to our target functions. We use the open-source implementation of the existing algorithms.

For the optimization of PPR fairness, we select 10\% of the nodes randomly. Then, for each selected node, we perform $b=50$ edge rewirings as suggested by each of the considered algorithms in an iterative manner.  We report the Wasserstein distances per round between the $S$ PPR ratio of the nodes in $S$ (i.e., $\{\mpi_i(S)|i\in S\}$) and the $S$ PPR ratio of the nodes in $T$ (i.e., $\{\mpi_i(S)|i\in T\}$).

\begin{sidewaystable}
    \centering
    \setlength{\tabcolsep}{1.5pt} 
    \caption{The running time (seconds, $s$) and the relative error of \textsc{Exact}, MFREC, and \textsc{Fast} on real-world networks. We denote the number of nodes and edges by $n$ and $m$, respectively. Parameter $r(\texttt{dis})$ and $\mpi(\texttt{dis})$ represent the population proportion of the disadvantaged group and its initial PageRank allocation, respectively.}
    \label{tab:data}
    \begin{tabular}{lrrccrrrrrrr}
        \toprule
        Networks & $n$ & $m$ & $r(\texttt{dis})$ & $\mpi(\texttt{dis})$ &
        \textsc{Exact} & MFREC & $\psi=1000$ & $\psi=2000$ &
        MFREC & $\psi=1000$ & $\psi=2000$ \\
        \midrule
        Books & 92 & 748 & 0.533 & 0.526 & 2.54 & 16.5 & 1.11 & 1.82 & 2.31 & 1.12 & 0.54 \\
        Blogs & 1,222 & 16,717 & 0.485 & 0.332 & 312.2 & 254 & 3.41 & 5.19 & 3.27 & 0.91 & 0.64 \\
        DBLP-Pub & 16,501 & 133,226 & 0.080 & 0.061 & 75189 & 41258 & 20.12 & 35.33 & 2.41 & 1.87 & 0.87 \\
        DBLP-Gender & 16,501 & 66,613 & 0.257 & 0.249 & 67418 & 64127 & 23.26 & 41.72 & 2.71 & 0.84 & 0.61 \\
        DBLP-Aminer & 423,469 & 1,231,211 & 0.185 & 0.175 & -- & -- & 114.54 & 189.32 & -- & -- & -- \\
        LinkedIn & 3,209,448 & 6,508,228 & 0.626 & 0.625 & -- & -- & 264.19 & 471.72 & -- & -- & -- \\
        \bottomrule
    \end{tabular}
\end{sidewaystable}

\subsection{Effectiveness Evaluation}
\begin{figure}[t!]
    \centering
    \includegraphics[width=0.8\columnwidth]{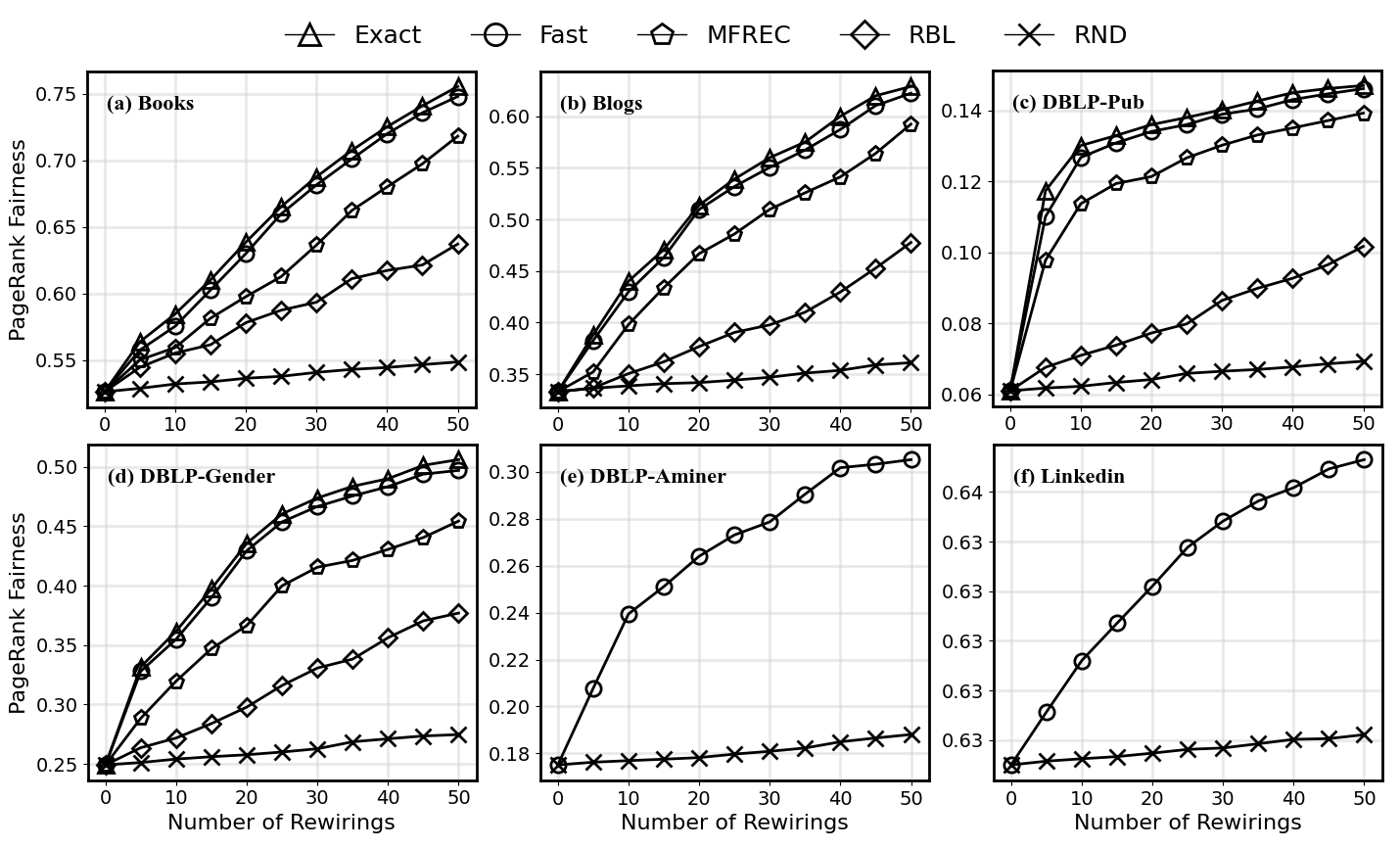}
    \caption{Comparison of different algorithms in increasing the PR fairness of the disadvantaged group.}
    \label{fig:prfair}
\end{figure}

\begin{figure}[t!]
    \centering
    \includegraphics[width=0.8\columnwidth]{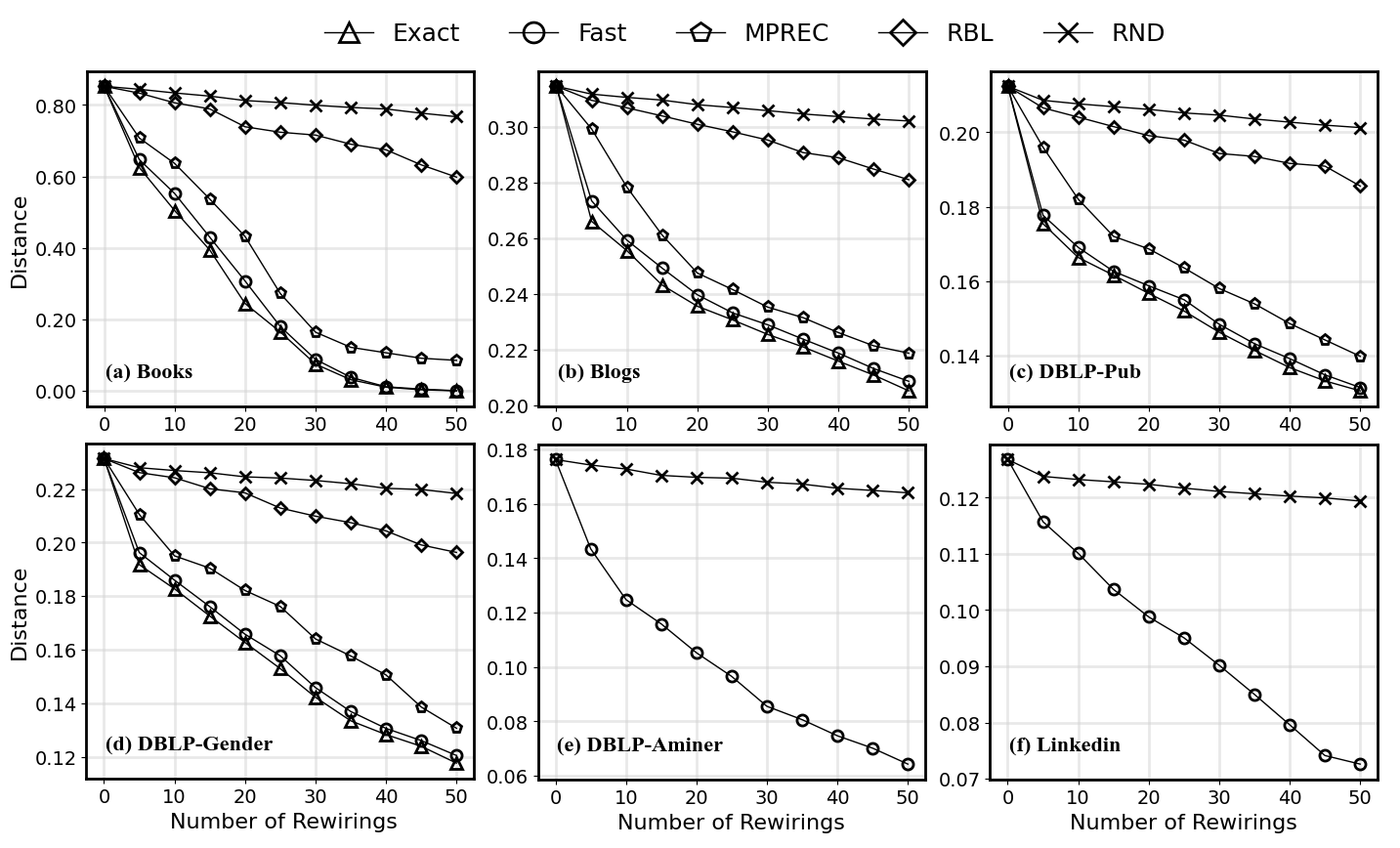}
    \caption{Wasserstein distances between the PPR ratio of the nodes in $S$ and the PPR ratio of the nodes in $T$.}
    \label{fig:distance}
\end{figure}

In this section, we evaluate the effectiveness of our algorithms against others. In Figure~\ref{fig:prfair}, we present the outcomes of our experiments on 6 real-world graphs. On each graph, the performance of every algorithm is evaluated along 50 rewiring operations. On all the graphs, it is clear that our algorithms outperform all other algorithms consistently, i.e., they achieve the best improvement in PageRank fairness. (The results are not reported for some algorithms on the larger networks since they didn't finish in a ``reasonable'' amount of time.)

In the last three columns of Table~\ref{tab:data}, we present the relative errors of the final PR allocation to the target group of \textsc{Fast} and \textsc{MFREC} compared to \textsc{Exact} for different networks and various numbers $\psi$ of samples. The results show that, the relative errors for \textsc{Fast} are consistently smaller with the smallest one being $0.54\%$. Additionally, as we increase $\psi$, the relative error decreases significantly. This again indicates that the results returned by \textsc{Fast} are very close to those corresponding to \textsc{Exact}. 

Now, let us look at the effectiveness of our algorithms in optimizing PPR fairness.
In Figure~\ref{fig:distance}, we plot the Wasserstein distance between the distribution of the percentage of $\mpi_i(S)$ for the nodes in $S$ and the nodes in $T=V\setminus S$. As shown in Figure~\ref{fig:distance}, our algorithms reduce the distance between these two distributions much more significantly.

\subsection{Efficiency Evaluation}

In the previous section, we observed that our algorithms consistently outperform other algorithms. Here, we focus on the run time analysis of these algorithms. To this end, we compare the proposed algorithms, namely \textsc{Exact} and \textsc{Fast}, with MFREC. (RBL and RND are excluded due to their unsatisfactory performance, see Figure~\ref{fig:prfair}.) We report the running time of each algorithm for rewiring 50 edges in different networks in Table~\ref{tab:data}. For algorithm \textsc{Fast}, we set the number of spanning rooted forests to be sampled as $\psi=1000$ and $\psi=2000$. The results show that \textsc{Fast} consistently runs faster than other algorithms. Specifically, for large-scale networks such as DBLP-Aminer and LinkedIn, all other algorithms do not terminate in one day, while \textsc{Fast} outputs the solution within 500 seconds even for the LinkedIn network with more than three million nodes. We note that the adapted existing algorithm MFREC can also handle the middle-scale networks DBLP-Pub and DBLP-Gender, but the running time is far more than \textsc{Fast}. Therefore, algorithm \textsc{Fast} is both effective and efficient, and scales to massive graphs.

\subsection{Impact of Reducing Candidates}
In the previous sections, we observed that, in addition to possessing theoretical guarantees, our algorithms outperform the existing strategies on different real-world graph data. To further validate the
robustness of our approach, we aim to understand if opting for 
$\widehat{\Delta}$ over $\Delta$, which facilitates the scalability of our algorithms, compromises the results’ precision. To shed light on this, we examined the interplay of $\widehat{\Delta}$ and $\Delta$
using two smaller-scale graphs, Books and Blogs, facilitating exact computations. As illustrated in Figure~\ref{fig:corre}, $\Delta$ and $\widehat{\Delta}$ are almost perfectly correlated, under Pearson correlation and as well as Spearman rank correlation. Thus, the outcome of these experiments support our reliance on $\widehat{\Delta}$, rather than $\Delta$ for our greedy rewiring strategies.

\begin{figure}[t!]
    \centering
    \includegraphics[width=0.8\columnwidth]{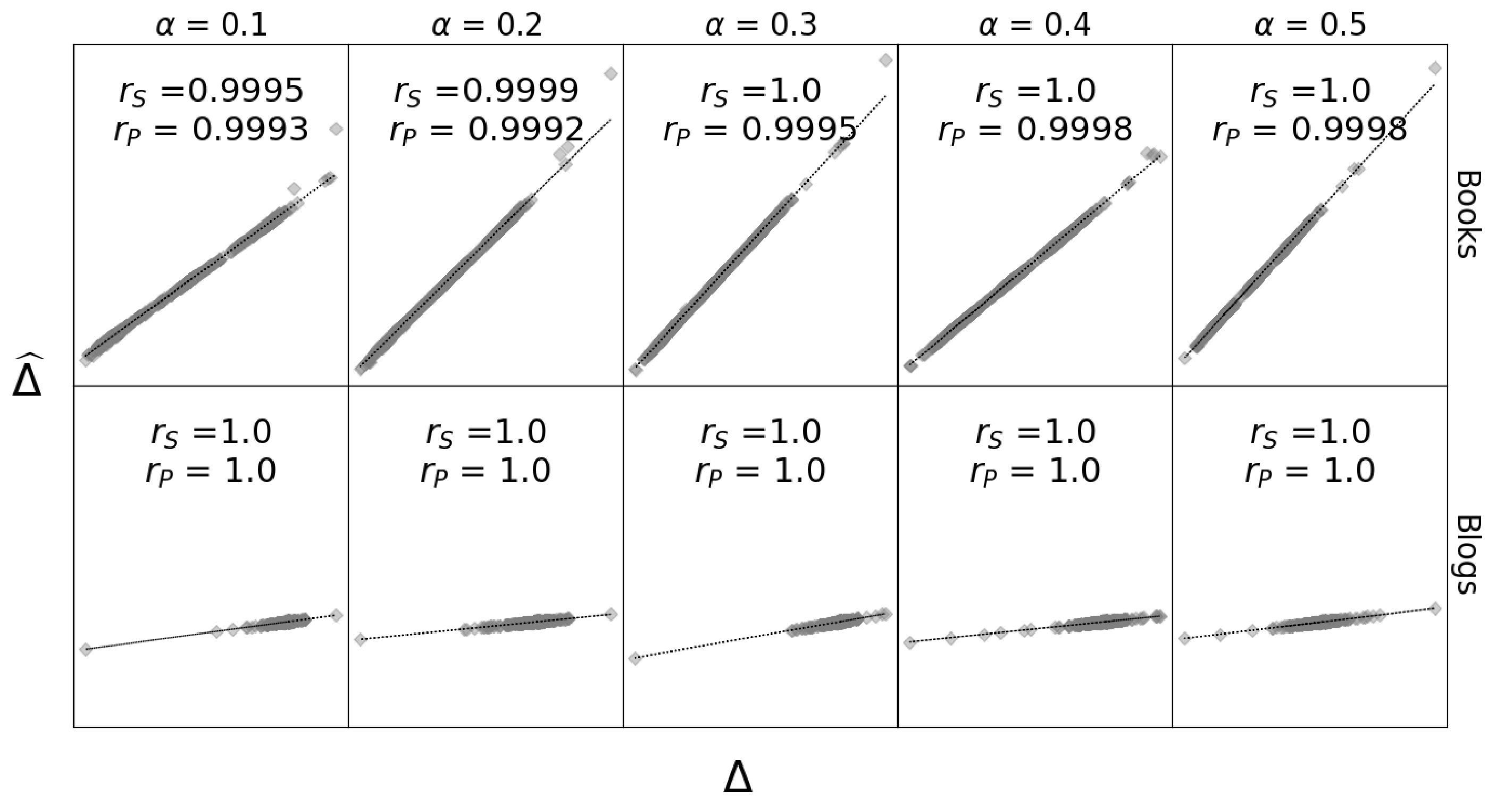}
    \caption{Correlation betweeen $\Delta$ and $\widehat{\Delta}$ for varying $\alpha$. $r_s$ and $r_p$ represent the Spearman and the Pearson coefficients, respectively.}
    \label{fig:corre}
\end{figure}

\section{Related Work}
\label{sec:related-work}

\textbf{Algorithmic Fairness.} Recently, there has been an increasing interest in algorithmic fairness in the context of social networks. Fairness is usually regarded as the lack of discrimination on the basis of some protective attribute, such as gender or race, and various mathematical formulations of fairness have been proposed~\cite{Kerchove2007MaximizingPV,EspinNoboa2021InequalityAI,PRbeyond,reducing2017,Kang2020InFoRMIF,recfairwww23,Kang2020InFoRMIF,www23fair}. Furthermore, several strategies have been introduced to maximize fairness~\cite{pitoura2021fairness,friedler2019comparative}, which can be categorized into \textit{pre-processing}, \textit{in-processing}, and \textit{post-processing}. These three improve fairness through adjusting the input, algorithm, and output, respectively. Our work falls under the umbrella of pre-processing strategies that modify the input network structure. Please refer to~\cite{friedler2019comparative,mehrabi2021survey,dong2023fairness,dai2022comprehensive} for more details on algorithmic fairness. 

\textbf{Network-based Fairness.} The evidences of inequality in graph-based characteristics of disadvantaged and advantaged groups have been consistently observed in real-world networks~\cite{friedler2019comparative,LiXiAh25}. Consequently, homophily and size imbalance have been investigated as parameters which can affect the ``social capital'' that different groups receive~\cite{dwork2012fairness, Feldman2014CertifyingAR}. It was shown that the combination of network characteristics such as homophily, preferential attachment and imbalances in group sizes leads to uneven degree distributions between groups~\cite{Amelkin2019FightingOC}. Furthermore, since the network data is commonly used as input of various social network algorithms, such as link recommenders, this inequality is intensified further~\cite{Karimi2018HomophilyIR}. 


\textbf{Graph Operations.} 
Different graph operations, such as edge deletion, addition, and rewiring, have been utilized to boost fairness~\cite{tsioutsiouliklis2022link, tsioutsiouliklis2021fairness,jalali2020information,bashardoust2022reducing}. Due to some of its desired properties, rewiring operation, in particular, has attracted significant attention in various fields such as opinion dynamics~\cite{chitra2020analyzing,bhalla2023local,santos2021link}, recommendation systems~\cite{Fabbri2022RewiringWR,Coupette2023ReducingET}, graph learning~\cite{ma2021graph,diffwire22} and others~\cite{valente2012network,d2019rewiring,kim2012network,han2019measuring,Chan2016OptimizingNR}. Specifically, its impact on overall network structure is minimal since it conserves the number of nodes, edges, and average degree, and influences some fundamental algebraic properties less substantially~\cite{ma2021graph}. In particular, for recommendation systems, the number of items recommended to an individual is typically fixed, necessitating only edge rewirings~\cite{Fabbri2022RewiringWR}.

\textbf{PageRank Fairness.} Our work studies a notion of fairness defined based on the very fundamental concept of PageRank. We aim to enhance fairness for a disadvantaged group by adjusting the graph structure. Our work can be seen as a continuation of the results from some recent works~\cite{tsioutsiouliklis2022link, tsioutsiouliklis2021fairness}, where modifying operation is adding/deleting edges.
However, our work distinguishes itself from theirs in three key aspects: 1) We introduce edge rewiring, which as discussed, possesses several desired preserving properties; 2) Their link recommendation restricts additions to a specified node, whereas we don't have this limitation; 3) Instead of simulating absorbing random walks for approximation, we harness the extended Wilson algorithm~\cite{Wilson1996GeneratingRS} for sampling, leading to more efficient algorithms. Notably, with small adaptation, our algorithms could also address the edge addition/deletion challenges as in~\cite{tsioutsiouliklis2022link}. Please refer to more details in the Appendix. Another relevant avenue of research aims at individual PageRank fairness~\cite{Kang2020InFoRMIF}, ensuring like individuals garner comparable PageRank values. Prior efforts also spotlighted solely on optimizing the PageRank of specific node subsets through edge operations~\cite{avrachenkov2006effect, csaji2014pagerank} or tracking Personalized PageRank in dynamic graphs~\cite{PRwww23}, but without a fairness lens.

\textbf{Algorithmic Techniques.}
The Wilson algorithm~\cite{Wilson1996GeneratingRS}, used in our work, has significant applications in multiple fields such as opinion dynamics~\cite{sun2023opinion}, the computation of resistance distances~\cite{Liao2023EfficientRD,Madry2015FastGO}, and graph signal processing~\cite{PiAmBaTr21}. However, we are the first to establish a connection between the Wilson algorithm and algorithmic fairness.
\section{Conclusion}\label{sec:conclusion}
We studied PageRank fairness and derived analytical formulas to quantify the effect of edge rewirings on PageRank and personalized PageRank fairness. For the problem of maximizing fairness towards a disadvantaged group, we proposed efficient liner-time edge rewiring algorithms which not only possess theoretical guarantees, but also consistently outperform the existing algorithms on different real-world datasets. In particular, our \textsc{Fast} algorithm can find accurate solutions for networks of million nodes in a few minutes.

There are several potential avenues for future research. In particular, we plan to explore other PageRank related fairness measures such as rank fairness~\cite{pitoura2021fairness} which aims to make the individuals in different groups have fair PageRank score rank in the network. Another interesting prospective research direction is to leverage novel approximation techniques to speed up our proposed algorithms further. This is essential considering the rapid growth of networks in the real world. Additionally, in the future works, we plan to explore alternative
scenarios, such as maximizing fairness when there are more than two groups and more than one disadvantaged group in the network.

\section*{Declaration of competing interest}
The authors declare that they have no known competing financial interests or personal relationships that could have appeared to influence the work reported in this paper.

\section*{Acknowledgements}

This work was supported  in part by the National Natural Science Foundation of China (Nos. 62372112 and 61872093).
\section*{CRediT authorship contribution statement}
\textbf{Changan Liu:} Writing -- review \& editing, Writing -- original draft, Visualization, Validation, Software, Resources, Methodology,
Investigation, Formal analysis, Data curation. \textbf{Haoxin Sun:} Writing -- review \& editing, Writing -- original draft, Visualization, Validation, Methodology,
Investigation, Data curation. \textbf{Ahad N. Zehmakan:} Writing -- review \& editing, Writing -- original draft, Visualization, Validation. \textbf{Zhongzhi Zhang:} Writing -- review \& editing, Writing -- original draft, Visualization,
Validation, Supervision, Project administration, Methodology, Funding acquisition, Conceptualization.

\bibliographystyle{elsarticle-num}
\normalem

\begin{thebibliography}{99}

\bibitem{d2021mitigating}
D'Angelo, G. and Abouei~Mehrizi, M. Mitigating negative influence diffusion is hard.
\newblock {\em Proceedings of the 30th ACM International Conference on Information \& Knowledge Management},  pp. 332--341.

\bibitem{becker2021influence}
Becker, R., d'Angelo, G., and Gilbert, H. Influence maximization with co-existing seeds.
\newblock {\em Proceedings of the 30th ACM International Conference on Information \& Knowledge Management},  pp. 100--109.

\bibitem{styczen2022targeted}
Styczen, M., Chen, B.-J., Teng, Y.-W., Pignolet, Y.-A., Chen, L., and Yang, D.-N. Targeted influence with community and gender-aware seeding.
\newblock {\em Proceedings of the 31st ACM International Conference on Information \& Knowledge Management},  pp. 4515--4519.

\bibitem{cong2023fairsample}
Cong, Z., Shi, B., Li, S., Yang, J., He, Q., and Pei, J. Fairsample: Training fair and accurate graph convolutional neural networks efficiently.
\newblock {\em IEEE Transactions on Knowledge and Data Engineering}, {\bf  36}, 1537--1551.

\bibitem{tkdeIMfair23}
Ali, J., Babaei, M., Chakraborty, A., Mirzasoleiman, B., Gummadi, K.~P., and Singla, A. On the fairness of time-critical influence maximization in social networks.
\newblock {\em IEEE Transactions on Knowledge and Data Engineering}, {\bf  35}, 2875--2886.

\bibitem{zhang2020fairness}
Zhang, T., Zhu, T., Li, J., Han, M., Zhou, W., and Philip, S.~Y. Fairness in semi-supervised learning: Unlabeled data help to reduce discrimination.
\newblock {\em IEEE Transactions on Knowledge and Data Engineering}, {\bf  34}, 1763--1774.

\bibitem{recfairwww23}
Chen, L., Wu, L., Zhang, K., Hong, R., Lian, D., Zhang, Z., Zhou, J., and Wang, M. Improving recommendation fairness via data augmentation.
\newblock {\em Proceedings of the Web Conference},  pp. 1012--1020. {ACM}.

\bibitem{fish2019gaps}
Fish, B., Bashardoust, A., Boyd, D., Friedler, S., Scheidegger, C., and Venkatasubramanian, S. Gaps in information access in social networks?
\newblock {\em The World Wide Web Conference},  pp. 480--490. ACM.

\bibitem{teng2021influencing}
Teng, Y.-W., Chen, H.-W., Yang, D.-N., Pignolet, Y.-A., Li, T.-W., and Chen, L. On influencing the influential: Disparity seeding.
\newblock {\em Proceedings of the 30th ACM International Conference on Information \& Knowledge Management},  pp. 1804--1813.

\bibitem{tsang2019group}
Tsang, A., Wilder, B., Rice, E., Tambe, M., and Zick, Y. Group-fairness in influence maximization.
\newblock {\em Proceedings of the 28th International Joint Conference on Artificial Intelligence},  pp. 5997--6005. AAAI Press.

\bibitem{yaseen2016influence}
Yaseen, Z.~K. and Marwan, Y. The influence of social media on recruitment and selection process in smes.
\newblock {\em Journal of Small Business and Entrepreneurship Development}, {\bf  4}, 21--27.

\bibitem{speicher2018potential}
Speicher, T., Ali, M., Venkatadri, G., Ribeiro, F.~N., Arvanitakis, G., Benevenuto, F., Gummadi, K.~P., Loiseau, P., and Mislove, A. Potential for discrimination in online targeted advertising.
\newblock {\em Conference on Fairness, Accountability and Transparency},  pp. 5--19. PMLR.

\bibitem{jalali2020information}
Jalali, Z.~S., Wang, W., Kim, M., Raghavan, H., and Soundarajan, S. On the information unfairness of social networks.
\newblock {\em Proceedings of the SIAM International Conference on Data Mining},  pp. 613--521. SIAM.

\bibitem{stoica2020seeding}
Stoica, A.-A., Han, J.~X., and Chaintreau, A. Seeding network influence in biased networks and the benefits of diversity.
\newblock {\em Proceedings of The Web Conference 2020},  pp. 2089--2098. ACM.

\bibitem{bashardoust2022reducing}
Bashardoust, A., Friedler, S.~A., Scheidegger, C.~E., Sullivan, B.~D., and Venkatasubramanian, S. Reducing access disparities in networks using edge augmentation.
\newblock {\em Proceedings of the ACM Conference on Fairness, Accountability, and Transparency},  pp. 1635--1651.

\bibitem{dwork2012fairness}
Dwork, C., Hardt, M., Pitassi, T., Reingold, O., and Zemel, R. Fairness through awareness.
\newblock {\em Proceedings of the 3rd Innovations in Theoretical Computer Science Conference},  pp. 214--226.

\bibitem{Feldman2014CertifyingAR}
Feldman, M., Friedler, S.~A., Moeller, J., Scheidegger, C.~E., and Venkatasubramanian, S. Certifying and removing disparate impact.
\newblock {\em Proceedings of the 21th ACM SIGKDD International Conference on Knowledge Discovery \& Data Mining},  pp. 259--268.

\bibitem{tsioutsiouliklis2021fairness}
Tsioutsiouliklis, S., Pitoura, E., Tsaparas, P., Kleftakis, I., and Mamoulis, N. Fairness-aware pagerank.
\newblock {\em Proceedings of the Web Conference},  pp. 3815--3826. ACM.

\bibitem{tsioutsiouliklis2022link}
Tsioutsiouliklis, S., Pitoura, E., Semertzidis, K., and Tsaparas, P. Link recommendations for \text{PageRank} fairness.
\newblock {\em Proceedings of the Web Conference},  pp. 3541--3551. ACM.

\bibitem{prfairwww24}
Stoica, A.-A., Litvak, N., and Chaintreau, A. Fairness rising from the ranks: Hits and pagerank on homophilic networks.
\newblock {\em Proceedings of the ACM on Web Conference 2024},  pp. 2594--2602.

\bibitem{Anatomy1998}
Brin, S. and Page, L. The anatomy of a large-scale hypertextual web search engine.
\newblock {\em Computer Networks}, {\bf  30}, 107--117.

\bibitem{PRbeyond}
Gleich, D.~F. Pagerank beyond the web.
\newblock {\em SIAM Review}, {\bf  57}, 321--363.

\bibitem{PRwww23}
Li, Z., Fu, D., and He, J. Everything evolves in personalized pagerank.
\newblock In Ding, Y., Tang, J., Sequeda, J.~F., Aroyo, L., Castillo, C., and Houben, G. (eds.), {\em Proceedings of the Web Conference},  pp. 3342--3352. {ACM}.

\bibitem{Chung2010PageRankAR}
Chung, F. R.~K. and Zhao, W. PageRank and Random Walks on Graphs. Springer,  Berlin, Heidelberg.

\bibitem{EspinNoboa2021InequalityAI}
Noboa, L.~E., Wagner, C., Strohmaier, M., and Karimi, F. Inequality and inequity in network-based ranking and recommendation algorithms.
\newblock {\em Scientific Reports}, {\bf  12}.

\bibitem{jeh2003scaling}
Jeh, G. and Widom, J. Scaling personalized web search.
\newblock {\em Proceedings of the 12th International Conference on World Wide Web},  pp. 271--279. ACM.

\bibitem{Yin2019ScalableGE}
Yin, Y. and Wei, Z. Scalable graph embeddings via sparse transpose proximities.
\newblock {\em Proceedings of the 25th ACM SIGKDD International Conference on Knowledge Discovery \& Data Mining},  pp. 1429--1437.

\bibitem{Chen2020ScalableGN}
Chen, M., Wei, Z., Ding, B., Li, Y., Yuan, Y., Du, X., and Wen, J.-R. Scalable graph neural networks via bidirectional propagation.
\newblock {\em Proceedings of the 34th International Conference on Neural Information Processing Systems}.

\bibitem{Wilson1996GeneratingRS}
Wilson, D.~B. Generating random spanning trees more quickly than the cover time.
\newblock {\em Proceedings of the Twenty-Eighth Annual ACM Symposium on Theory of Computing},  pp. 296--303.

\bibitem{fairdrop21}
Spinelli, I., Scardapane, S., Hussain, A., and Uncini, A. Fairdrop: Biased edge dropout for enhancing fairness in graph representation learning.
\newblock {\em IEEE Transactions on Artificial Intelligence}, {\bf  3}, 344--354.

\bibitem{li2021on}
Li, P., Wang, Y., Zhao, H., Hong, P., and Liu, H. On dyadic fairness: Exploring and mitigating bias in graph connections.
\newblock {\em International Conference on Learning Representations}.

\bibitem{LiZhAh24}
Liu, Changan and Zhou, Xiaotian and Zehmakan, Ahad N. and Zhang, Zhongzhi.{A Fast Algorithm for Moderating Critical Nodes via Edge Removal}. \newblock{\em IEEE Transactions on Knowledge and Data Engineering}, {\bf 36}, 2024, 1385--1398.

\bibitem{LiXiAh25}
Changan Liu and Xiaotian Zhou and Ahad N. Zehmakan and Zhongzhi Zhang. Promoting Fairness in Information Access within Social Networks. \newblock{\em arXiv}, {2512.14711}, 2025.

\bibitem{burst_2020}
Masrour, F., Wilson, T., Yan, H., Tan, P.-N., and Esfahanian, A. Bursting the filter bubble: Fairness-aware network link prediction.
\newblock {\em Proceedings of the AAAI Conference on Artificial Intelligence}, {\bf  34}, 841--848.

\bibitem{swift2022maximizing}
Swift, I.~P., Ebrahimi, S., Nova, A., and Asudeh, A. Maximizing fair content spread via edge suggestion in social networks.
\newblock {\em Proceedings of the VLDB Endowment}, {\bf  15}, 2692–2705.

\bibitem{chitra2020analyzing}
Chitra, U. and Musco, C. Analyzing the impact of filter bubbles on social network polarization.
\newblock {\em Proceedings of the 13th International Conference on Web Search and Data Mining},  pp. 115--123.

\bibitem{bhalla2023local}
Bhalla, N., Lechowicz, A., and Musco, C. Local edge dynamics and opinion polarization.
\newblock {\em Proceedings of the Sixteenth ACM International Conference on Web Search and Data Mining},  pp. 6--14.

\bibitem{santos2021link}
Santos, F.~P., Lelkes, Y., and Levin, S.~A. Link recommendation algorithms and dynamics of polarization in online social networks.
\newblock {\em Proceedings of the National Academy of Sciences}, {\bf  118}, e2102141118.

\bibitem{Fabbri2022RewiringWR}
Fabbri, F., Wang, Y., Bonchi, F., Castillo, C., and Mathioudakis, M. Rewiring what-to-watch-next recommendations to reduce radicalization pathways.
\newblock {\em Proceedings of the ACM Web Conference},  pp. 2719--2728.

\bibitem{Coupette2023ReducingET}
Coupette, C., Neumann, S., and Gionis, A. Reducing exposure to harmful content via graph rewiring.
\newblock {\em Proceedings of the 29th ACM SIGKDD International Conference on Knowledge Discovery \& Data Mining},  pp. 323--334.

\bibitem{ma2021graph}
Ma, Y., Wang, S., Derr, T., Wu, L., and Tang, J. Graph adversarial attack via rewiring.
\newblock {\em Proceedings of the 27th ACM SIGKDD International Conference on Knowledge Discovery \& Data Mining},  pp. 1161--1169.

\bibitem{diffwire22}
Arnaiz-Rodr{\'\i}guez, A., Begga, A., Escolano, F., and Oliver, N.~M. Diffwire: Inductive graph rewiring via the lov\'asz bound.
\newblock {\em The First Learning on Graphs Conference}.

\bibitem{valente2012network}
Valente, T.~W. Network interventions.
\newblock {\em Science}, {\bf  337}, 49--53.

\bibitem{d2019rewiring}
D'Alelio, D., Hay~Mele, B., Libralato, S., Ribera~d'Alcal{\`a}, M., and Jord{\'a}n, F. Rewiring and indirect effects underpin modularity reshuffling in a marine food web under environmental shifts.
\newblock {\em Ecology and Evolution}, {\bf  9}, 11631--11646.

\bibitem{kim2012network}
Kim, J., Kim, I., Han, S.~K., Bowie, J.~U., and Kim, S. Network rewiring is an important mechanism of gene essentiality change.
\newblock {\em Scientific Reports}, {\bf  2}, 900.

\bibitem{han2019measuring}
Han, Y. and Goetz, S.~J. Measuring network rewiring over time.
\newblock {\em Plos One}, {\bf  14}, e0220295.

\bibitem{Chan2016OptimizingNR}
Chan, H. and Akoglu, L. Optimizing network robustness by edge rewiring: a general framework.
\newblock {\em Data Mining and Knowledge Discovery}, {\bf  30}, 1395--1425.

\bibitem{rastegarpanah2019fighting}
Rastegarpanah, B., Gummadi, K.~P., and Crovella, M. Fighting fire with fire: Using antidote data to improve polarization and fairness of recommender systems.
\newblock {\em Proceedings of the 25th ACM International Conference on Web Search and Data Mining},  pp. 231--239.

\bibitem{wang2019enhancing}
Wang, Q., Yin, H., Wang, H., Nguyen, Q. V.~H., Huang, Z., and Cui, L. Enhancing collaborative filtering with generative augmentation.
\newblock {\em Proceedings of the 25th ACM SIGKDD International Conference on Knowledge Discovery and Data Mining},  pp. 548--556.

\bibitem{ChSh97}
Chebotarev, P.~Y. and Shamis, E.~V. The matrix-forest theorem and measuring relations in small social groups.
\newblock {\em Automation and Remote Control}, {\bf  58}, 1505--1514.

\bibitem{ChSh98}
Chebotarev, P.~Y. and Shamis, E.~V. On proximity measures for graph vertices.
\newblock {\em Automation and Remote Control}, {\bf  59}, 1443--1459.

\bibitem{Chebotarev2006SpanningFA}
Chebotarev, P.~Y. Spanning forests and the golden ratio.
\newblock {\em Discrete Applied Mathematics}, {\bf  156}, 813--821.

\bibitem{Golender1981GraphPM}
Golender, V.~E., Drboglav, V.~V., and Rosenblit, A.~B. Graph potentials method and its application for chemical information processing.
\newblock {\em Journal of Chemical Information and Computer Science}, {\bf  21}, 196--204.

\bibitem{Ch82}
Chaiken, S. A combinatorial proof of the all minors matrix tree theorem.
\newblock {\em SIAM Journal on Algebraic Discrete Methods}, {\bf  3}, 319--329.

\bibitem{Merris1998DoublySG}
Merris, R. Doubly stochastic graph matrices, ii.
\newblock {\em Linear \& Multilinear Algebra}, {\bf  45}, 275--285.

\bibitem{Fouss2007RandomWalkCO}
Fouss, F., Pirotte, A., Renders, J.-M., and Saerens, M. Random-walk computation of similarities between nodes of a graph with application to collaborative recommendation.
\newblock {\em IEEE Transactions on Knowledge and Data Engineering}, {\bf  19}, 355--369.

\bibitem{Jin2019ForestDC}
Jin, Y., Bao, Q., and Zhang, Z. Forest distance closeness centrality in disconnected graphs.
\newblock {\em IEEE International Conference on Data Mining},  pp. 339--348.

\bibitem{Senelle2013TheSD}
Senelle, M., Garc{\'i}a-D{\'i}ez, S., Mantrach, A., Shimbo, M., Saerens, M., and Fouss, F. The sum-over-forests density index: Identifying dense regions in a graph.
\newblock {\em IEEE Transactions on Pattern Analysis and Machine Intelligence}, {\bf  36}, 1268--1274.

\bibitem{sun2023opinion}
Sun, H. and Zhang, Z. Opinion optimization in directed social networks.
\newblock {\em Proceedings of the AAAI Conference on Artificial Intelligence},  pp. 4623--4632.

\bibitem{fu2022disco}
Fu, D., Ban, Y., Tong, H., Maciejewski, R., and He, J. Disco: comprehensive and explainable disinformation detection.
\newblock {\em Proceedings of the 31st ACM International Conference on Information \& Knowledge Management},  pp. 4848--4852.

\bibitem{fu2021sdg}
Fu, D. and He, J. Sdg: A simplified and dynamic graph neural network.
\newblock {\em Proceedings of the 44th International ACM SIGIR Conference on Research and Development in Information Retrieval},  pp. 2273--2277.

\bibitem{kamvar2003extrapolation}
Kamvar, S.~D., Haveliwala, T.~H., Manning, C.~D., and Golub, G.~H. Extrapolation methods for accelerating pagerank computations.
\newblock {\em Proceedings of the 12th International Conference on World Wide Web},  pp. 261--270.

\bibitem{Tong2006FastRW}
Tong, H., Faloutsos, C., and Pan, J.-Y. Fast random walk with restart and its applications.
\newblock {\em Sixth International Conference on Data Mining},  pp. 613--622.

\bibitem{yao2012anomaly}
Yao, Z., Mark, P., and Rabbat, M. Anomaly detection using proximity graph and pagerank algorithm.
\newblock {\em IEEE Transactions on Information Forensics and Security}, {\bf  7}, 1288--1300.

\bibitem{yoon2019fast}
Yoon, M., Hooi, B., Shin, K., and Faloutsos, C. Fast and accurate anomaly detection in dynamic graphs with a two-pronged approach.
\newblock {\em Proceedings of the 25th ACM SIGKDD International Conference on Knowledge Discovery \& Data Mining},  pp. 647--657.

\bibitem{Me73}
Meyer, C.~D., Jr Generalized inversion of modified matrices.
\newblock {\em SIAM Journal on Applied Mathematics}, {\bf  24}, 315--323.

\bibitem{Wi96}
Wilson, D.~B. Generating random spanning trees more quickly than the cover time.
\newblock {\em Proceedings of the Twenty-Eighth Annual ACM Symposium on Theory of Computing}  296–303.

\bibitem{WiPr96}
Wilson, D.~B. and Propp, J.~G. How to get an exact sample from a generic markov chain and sample a random spanning tree from a directed graph, both within the cover time.
\newblock {\em Proceedings of the Seventh Annual ACM-SIAM Symposium on Discrete Algorithms},  pp. 448--457.

\bibitem{La80}
Lawler and Gregory, F. A self-avoiding random walk.
\newblock {\em Duke Mathematical Journal}, {\bf  47}, 655--693.

\bibitem{AvLuGaAl18}
Avena, L. and Gaudilli{\`e}re, A. Two applications of random spanning forests.
\newblock {\em Journal of Theoretical Probability}, {\bf  31}, 1975--2004.

\bibitem{PiAmBaTr21}
Pilavc{\i}, Y.~Y., Amblard, P.-O., Barthelme, S., and Tremblay, N. Graph {T}ikhonov regularization and interpolation via random spanning forests.
\newblock {\em IEEE Transactions on Signal and Information Processing over Networks}, {\bf  7}, 359--374.

\bibitem{Ma00}
Marchal, P. Loop-erased random walks, spanning trees and hamiltonian cycles.
\newblock {\em Electronic Communications in Probability}, {\bf  5}, 39--50.

\bibitem{Ho63}
Hoeffding and Wassily Probability inequalities for sums of bounded random variables.
\newblock {\em Journal of the American Statistical Association}, {\bf  58}, 13--30.

\bibitem{LeSo16}
Leskovec, J. and Sosi{\v{c}}, R. {SNAP}: A general-purpose network analysis and graph-mining library.
\newblock {\em ACM Transactions on Intelligent Systems and Technolog}, {\bf  8}, 1.

\bibitem{haddadan2021repbublik}
Haddadan, S., Menghini, C., Riondato, M., and Upfal, E. Repbublik: Reducing polarized bubble radius with link insertions.
\newblock {\em Proceedings of the 14th ACM International Conference on Web Search and Data Mining},  pp. 139--147.

\bibitem{Kerchove2007MaximizingPV}
{de Kerchove}, C., Ninove, L., and {van Dooren}, P. Maximizing pagerank via outlinks.
\newblock {\em Linear Algebra and its Applications}, {\bf  429}, 1254--1276.

\bibitem{reducing2017}
Garimella, K., De~Francisci~Morales, G., Gionis, A., and Mathioudakis, M. Reducing controversy by connecting opposing views.
\newblock {\em Proceedings of the Tenth ACM International Conference on Web Search and Data Mining},  pp. 81--90. ACM.

\bibitem{Kang2020InFoRMIF}
Kang, J., He, J., Maciejewski, R., and Tong, H. Inform: Individual fairness on graph mining.
\newblock {\em Proceedings of the 26th ACM SIGKDD International Conference on Knowledge Discovery \& Data Mining}.

\bibitem{www23fair}
Fu, D., Zhou, D., Maciejewski, R., Croitoru, A., Boyd, M., and He, J. Fairness-aware clique-preserving spectral clustering of temporal graphs.
\newblock {\em Proceedings of the Web Conference},  pp. 3755--3765. {ACM}.

\bibitem{pitoura2021fairness}
Pitoura, E., Stefanidis, K., and Koutrika, G. Fairness in rankings and recommendations: An overview.
\newblock {\em The VLDB Journal}, {\bf  31}, 431–458.

\bibitem{friedler2019comparative}
Friedler, S.~A., Scheidegger, C., Venkatasubramanian, S., Choudhary, S., Hamilton, E.~P., and Roth, D. A comparative study of fairness-enhancing interventions in machine learning.
\newblock {\em Proceedings of the Conference on Fairness, Accountability, and Transparency},  pp. 329--338. ACM.

\bibitem{mehrabi2021survey}
Mehrabi, N., Morstatter, F., Saxena, N., Lerman, K., and Galstyan, A. A survey on bias and fairness in machine learning.
\newblock {\em ACM Computing Surveys}, {\bf  54}, 1--35.

\bibitem{dong2023fairness}
Dong, Y., Ma, J., Wang, S., Chen, C., and Li, J. Fairness in graph mining: A survey.
\newblock {\em IEEE Transactions on Knowledge and Data Engineering}, {\bf  35}, 10583--10602.

\bibitem{dai2022comprehensive}
Dai, E., Zhao, T., Zhu, H., Xu, J., Guo, Z., Liu, H., Tang, J., and Wang, S. A comprehensive survey on trustworthy graph neural networks: Privacy, robustness, fairness, and explainability.
\newblock {\em arXiv preprint arXiv:2204.08570}.

\bibitem{Amelkin2019FightingOC}
Amelkin, V. and Singh, A.~K. Fighting opinion control in social networks via link recommendation.
\newblock {\em Proceedings of the 25th ACM SIGKDD International Conference on Knowledge Discovery and Data Mining},  pp. 677--685.

\bibitem{Karimi2018HomophilyIR}
Karimi, F., G{\'e}nois, M., Wagner, C., Singer, P., and Strohmaier, M. Homophily influences ranking of minorities in social networks.
\newblock {\em Scientific Reports}, {\bf  8}.

\bibitem{avrachenkov2006effect}
Avrachenkov, K. and Litvak, N. The effect of new links on google pagerank.
\newblock {\em Stochastic Models}, {\bf  22}, 319--331.

\bibitem{csaji2014pagerank}
Cs{\'a}ji, B.~C., Jungers, R.~M., and Blondel, V.~D. Pagerank optimization by edge selection.
\newblock {\em Discrete Applied Mathematics}, {\bf  169}, 73--87.

\bibitem{Liao2023EfficientRD}
Liao, M., Li, R., Dai, Q., Chen, H., Qin, H., and Wang, G. Efficient resistance distance computation: The power of landmark-based approaches.
\newblock {\em Proceedings of the ACM on Management of Data},  pp. 1--27.

\bibitem{Madry2015FastGO}
Madry, A., Straszak, D., and Tarnawski, J. Fast generation of random spanning trees and the effective resistance metric.
\newblock {\em ACM-SIAM Symposium on Discrete Algorithms},  pp. 2019--2036.

\end{thebibliography}

\end{document}